\documentclass[11pt]{article}

\usepackage{amsmath}
\usepackage{amsthm}
\usepackage{amssymb}
\usepackage{graphics}
\usepackage{url}

\usepackage{amsfonts}
\usepackage{hhline}
\usepackage{multirow}
\usepackage{xcolor}
\usepackage{float}
\usepackage{bbm}
\usepackage{algorithm}
\usepackage{algpseudocode}
\usepackage{comment}
\usepackage{caption}
\usepackage{subcaption}
\usepackage{pgf}
\usepackage{tikz}
\usepackage{enumitem}
\usepackage[left=1in,right=1in,top=1in,bottom=1in]{geometry}
\usepackage{nicefrac}

\usepackage{float}
\usepackage[font=bf,skip=\baselineskip]{caption}
\usepackage{wrapfig}
\usepackage{soul}
\usepackage{times}

\usepackage{comment}
\usepackage{xcolor}

\usepackage{amsmath}
\usepackage{amsthm}
\usepackage{amssymb}
\usepackage{graphics}
\usepackage{url}
\usepackage{amsfonts}
\usepackage{hhline}
\usepackage{multirow}

\usepackage{float}
\usepackage{bbm}
\usepackage{algorithm}
\usepackage{algpseudocode}
\usepackage{caption}
\usepackage{subcaption}
\usepackage{pgf}
\usepackage{tikz}
\usepackage{enumitem}

\usepackage{natbib}

\allowdisplaybreaks

\newcommand{\V}{\mathcal{V}}

\newcommand{\OPT}{\mathcal{OPT}}

\let\cite\citet

\newtheorem{theorem}{Theorem}[section]
\newtheorem{lemma}[theorem]{Lemma}
\newtheorem{corollary}[theorem]{Corollary}

\newtheorem{fact}[theorem]{Fact}

\newtheorem{definition}[theorem]{Definition}

\begin{document}

\title{$k$-center Clustering under Perturbation Resilience
\footnote{
Authors' addresses: 
\texttt{ninamf@cs.cmu.edu, nika.haghtalab@microsoft.com, crwhite@cs.cmu.edu}. 
This paper combines and extends results appearing in ICALP 2016 \citep{symmetric}
and arXiv:1705.07157 \citep{bridging}.
}
}

\author{Maria-Florina Balcan \and Nika Haghtalab \and Colin White}

\date{}

\maketitle

\begin{abstract}

The $k$-center problem is a canonical and long-studied facility
location and clustering problem with many applications in both its
symmetric and asymmetric forms.  Both versions of the problem have
tight approximation factors on worst case instances: a
$2$-approximation for symmetric $k$-center and an
$O(\log^*(k))$-approximation for the asymmetric version. Therefore to
improve on these ratios, one must go beyond the worst case.  

In this work, we take this approach and provide strong positive results both for
the asymmetric and symmetric $k$-center problems under a natural
input stability (promise) condition called {\em $\alpha$-perturbation
resilience} \citep{bilu2012stable}, which states that the
optimal solution does not change under any $\alpha$-factor
perturbation to the input distances.
We provide algorithms that give strong guarantees simultaneously for stable and non-stable instances:
our algorithms always inherit the worst-case guarantees of clustering approximation algorithms,
and output the optimal solution if the input is 2-perturbation resilient.
In particular, we show that if the input is only perturbation resilient on part of the data, our algorithm will
return the optimal clusters from the region of the data that is perturbation resilient,
while achieving the best worst-case approximation guarantee on the remainder of the data.
Furthermore, we prove our result is tight by showing symmetric $k$-center under
$(2-\epsilon)$-perturbation resilience is hard unless $NP=RP$.

The impact of our results are multifaceted. First, 
to our knowledge, asymmetric $k$-center is the first problem that is hard to approximate to
any constant factor in the worst case, yet can be optimally solved in
polynomial time under perturbation resilience for a constant value of
$\alpha$.
This is also the first tight result for any problem under perturbation resilience,
i.e., this is the first time the exact value of $\alpha$ for which the problem
switches from being NP-hard to efficiently computable has been found.
Furthermore, our results illustrate a surprising relationship between symmetric and
asymmetric $k$-center instances under perturbation resilience. 
Unlike approximation ratio, for which symmetric $k$-center
is easily solved to a factor of $2$ but asymmetric $k$-center cannot
be approximated to any constant factor, both symmetric and asymmetric
$k$-center can be solved optimally under resilience to 2-perturbations.
Finally, our guarantees in the setting where only part of the data satisfies perturbation resilience
makes these algorithms more applicable to real-life instances.

\end{abstract}

\newpage

\tableofcontents
\newpage

\section{Introduction} 

Clustering is a fundamental problem in combinatorial optimization with 
a wide range of applications including bioinformatics, computer vision, text analysis, and countless others.
The underlying goal is to partition a given set of points to maximize similarity within a partition and
minimize similarity across different partitions.
A common approach to clustering is to consider an objective function over all possible partitionings 
and seek solutions that are optimal according to the objective.
Given a set of points (and a distance metric), common clustering objectives include finding
$k$ centers to minimize the sum of the distance from each point to its closest center ($k$-median),
or to minimize the maximum distance from a point to its closest center ($k$-center).

Traditionally, the theory of clustering (and more generally, the theory of algorithms) 
has focused on the analysis of worst-case instances 
\citep{arya2004local,byrka2015improved,charikar1999constant,outliers,chen2008constant,gonzalez1985clustering,makarychev2016bi}.
For example, it is well known the popular objective functions are 
provably NP-hard to optimize exactly or even approximately (APX-hard)~\citep{gonzalez1985clustering,jain2002new,lee2017improved},
so research has focused on finding approximation algorithms. 
While this perspective has led to many elegant approximation algorithms and lower bounds for worst-case instances, 
it is often overly pessimistic of an algorithm's performance on ``typical'' instances or real world instances.
A rapidly developing line of work in the algorithms community, 
the so-called \emph{beyond worst-case analysis} of  algorithms (BWCA),
considers the design and analysis of problem instances under natural structural properties which may be satisfied in real-world applications.
For example, the popular notion of $\alpha$-perturbation resilience, introduced by 
\cite{bilu2012stable},
considers instances such that the optimal solution does not change when
the input distances are allowed to increase by up to a factor of $\alpha$.
The goals of BWCA are twofold: 
\emph{(1)} to design new algorithms with strong performance guarantees under the added assumptions~\citep{as,hardt2013beyond,kumar2004simple,tim}, and
\emph{(2)} to prove strong guarantees under BWCA assumptions for existing algorithms used in practice~\citep{makarychev2014bilu,ostrovsky2012effectiveness,spielman2004smoothed}.
An example of goal \emph{(1)} is a series of work focused on finding exact algorithms for $k$-median, $k$-means, and $k$-center
clustering under $\alpha$-perturbation resilience
\citep{awasthi2012center,balcan2012clustering,angelidakis2017algorithms}. 
The goal in this line of work is to find the minimum value of $\alpha\geq 1$ 
for which optimal clusterings of $\alpha$-perturbation resilient instances can be found efficiently.
Two examples of goal \emph{(2)} are as follows.
Ostrovsky et al.\ showed that $k$-means++ outputs a near-optimal clustering, as long as the data satisfies a natural clusterability criterion \citep{ostrovsky2012effectiveness},
and \citet{spielman2004smoothed} established the expected runtime of the simplex method is $O(n)$ under smoothed analysis.


In approaches for answering goals \emph{(1)} and \emph{(2)},
researchers have developed an array of sophisticated tools exploiting the structural properties 
of such instances leading to algorithms which can output the optimal solution.
However, overly exploiting a BWCA assumption can lead to algorithms
that perform poorly when the input data does not exactly satisfy the given assumption.
Indeed, recent analyses and technical tools are susceptible to small deviations from the BWCA assumptions 
which can propagate when just a small fraction of the data does not satisfy the assumption.
For example, some recent algorithms make use of a dynamic programming subroutine which crucially need the entire instance
to satisfy the specific structure guaranteed by the BWCA assumption.
To continue the efforts of BWCA in bridging the theory-practice gap, it is essential to study algorithms whose guarantees degrade
gracefully to address scenarios that present mild deviations from the standard BWCA assumptions.
Another downside of existing approaches is that BWCA assumptions are often not efficiently verifiable.
This creates a catch-22 scenario: it is only useful to run the algorithms if the data satisfies certain assumptions,
but a user cannot check these assumptions efficiently.
For example, by nature of $\alpha$-perturbation resilience
(that the optimal clustering does not change under \emph{all} $\alpha$-perturbations
of the input), it is not known how to test this condition without computing the optimal clustering over $\Omega\left(2^n\right)$
different perturbations.
To alleviate these issues, in this work we also focus on what we propose should be a third goal for BWCA:
to show (new or existing) algorithms 
whose performance degrades gracefully on instances that only partially meet the BWCA assumptions.

\subsection{Our results and techniques}
In this work, we address goals \emph{(1)}, \emph{(2)}, and \emph{(3)} of BWCA by providing robust algorithms which give the optimal solution
under perturbation resilience, and also perform well when the data is partially perturbation resilient, or not at all perturbation resilient. 
These algorithms act as an interpolation between worst-case and beyond worst-case analysis.
We focus on the symmetric/asymmetric $k$-center objective under perturbation resilience.
Our algorithms simultaneously output the optimal clusters from the stable regions of the data, while achieving
state-of-the-art approximation ratios over the rest of the data.
In most cases, our algorithms are natural modifications to existing approximation algorithms, thus achieving goal \emph{(2)} of BWCA.
To achieve these two-part guarantees, we define the notion of perturbation resilience on a subset of the datapoints.
All prior work has only studied perturbation resilience as it applies to the entire dataset.
Informally, a subset $S'\subseteq S$ satisfies $\alpha$-perturbation resilience if all points $v\in S'$ remain in
the same optimal cluster under any $\alpha$-perturbation to the input.
We show that our algorithms return all optimal clusters from these locally stable regions.
Most of our results also apply under the recently defined, weaker condition of 
$\alpha$-\emph{metric perturbation resilience}~\citep{angelidakis2017algorithms}, 
which states that the optimal solution cannot change under
the metric closure of any $\alpha$-perturbation.
We list all our results in Table \ref{tab:results},
and give a summary of the results and techniques below.

\paragraph{$k$-center under 2-perturbation resilience}
In Section~\ref{sec:2pr}, 
we show that \emph{any} 2-approximation algorithm for $k$-center will always return
the clusters satisfying $2$-perturbation resilience.
Therefore, since there are well-known 2-approximation algorithms for symmetric $k$-center,
our analysis shows these will output the optimal clustering under 2-perturbation resilience.
For asymmetric $k$-center, we give a new algorithm which outputs the optimal clustering under 2-perturbation resilience.
It works
by first computing the ``symmetrized set'', or the points which demonstrate a rough symmetry.
We show how to optimally cluster the symmetrized set, and then we show how to add back the highly asymmetric
points into their correct clusters.

\paragraph{Hardness of symmetric $k$-center under $(2-\delta)$-perturbation resilience.}
In Section \ref{sec:hardness}, we prove there is no polynomial time algorithm for symmetric $k$-center under $(2-\delta)$-perturbation resilience unless $NP=RP$, which
shows that our perturbation resilience results are tight for both symmetric and asymmetric $k$-center. 
In particular, it implies that we have identified the exact moment ($\alpha=2$) where the problem switches
from efficiently computable to NP-hard, for both symmetric and asymmetric $k$-center.
For this hardness result, we use
a reduction from a variant of perfect dominating set. To show that this variant is itself hard, we construct
a chain of parsimonious reductions (reductions which conserve the number of solutions) from
3-dimensional matching to perfect dominating set.

Our upper bound for asymmetric $k$-center under 2-PR and lower bound for symmetric $k$-center under $(2-\delta)$-PR 
illustrate a surprising relationship between symmetric and asymmetric 
$k$-center instances under perturbation resilience. 
Unlike approximation ratio, for which symmetric $k$-center is easily solved to a factor of 2
but asymmetric $k$-center cannot be approximated to any constant factor, 
both symmetric and asymmetric $k$-center can be solved optimally under resilience to 2-perturbations. 
Overall, this is the first tight result
quantifying the power of perturbation resilience for a canonical combinatorial optimization problem.

\paragraph{Local perturbation resilience}
In Section~\ref{sec:local}, we apply our results from Section \ref{sec:2pr} to
the local perturbation resilience setting.
For symmetric $k$-center, we show that any 2-approximation outputs all optimal clusters from 2-perturbation resilient regions.
For asymmetric $k$-center, we design a new algorithm based off of the worst-case $O(\log^* n)$ approximation algorithm due to \cite{vishwanathan},
which is tight~\citep{chuzhoy2005asymmetric}. 
We give new insights into this algorithm, which allow us to show a modification of the algorithm
which outputs all optimal clusters from 2-perturbation resilient regions, while keeping the worst-case $O(\log^* n)$ guarantee overall.
If the entire dataset satisfies 2-perturbation resilience, then our algorithm outputs the optimal clustering.
We combine the tools of Vishwanathan with the perturbation resilience assumption to prove this two-part guarantee.
Specifically, we use the notion of a \emph{center-capturing vertex (CCV)}, which is used in the first phase of the approximation algorithm to pull
out supersets of clusters.
We show that each optimal center from a 2-perturbation resilient subset is a CCV and satisfies a separation property;
we prove this by carefully constructing a 2-perturbation in which points from other clusters cannot be too close to the center without causing a contradiction.
The structure allows us to modify the approximation algorithm of \cite{vishwanathan} to ensure that
optimal clusters from perturbation resilient subsets are pulled out separately in the first phase.
All of our guarantees hold under the weaker notion of metric perturbation resilience.

\paragraph{Efficient algorithms for symmetric and asymmetric $k$-center under $(3,\epsilon)$-perturbation resilience.}
In Section~\ref{sec:3eps}, we consider $(\alpha,\epsilon)$-perturbation resilience, 
which states that at most $\epsilon n$ total points can swap into or out of each cluster under any $\alpha$-perturbation.
For symmetric $k$-center, we show that any 2-approximation algorithm will return
the optimal clusters from $(3,\epsilon)$-perturbation resilient regions, assuming a mild lower bound on optimal cluster sizes,
and for asymmetric $k$-center, we give an algorithm which outputs a clustering that
is $\epsilon$-close to the optimal clustering.
Our main structural tool is showing that if any single point $v$ is 
close to an optimal cluster other than its own,
then $k-1$ centers achieve the optimal radius under a carefully constructed 3-perturbation.
Any other point we add to the set of centers
must create a clustering that is $\epsilon$-close to the optimal clustering, 
and we show all of these sets cannot simultaneously
be consistent with one another, thus causing a contradiction. A key concept in our analysis is defining the
notion of a cluster-capturing center, which allows us to reason about which points can capture a cluster when its center is removed.

\renewcommand{\arraystretch}{1.5}
\begin{table}[h]
\begin{center}
\begin{tabular}{|p{5.5cm}||p{2.5cm}|p{.5cm}||p{1cm}|p{1cm}|p{2cm}|}
\hline
\textbf{Problem}  & \textbf{Guarantee} & \textbf{$\alpha$} & \textbf{Metric}   & \textbf{Local} & \textbf{Theorem} \\ \hline\hline
Symmetric $k$-center under $\alpha$-PR  & $\OPT$  & $2$ & Yes & Yes & Theorem \ref{thm:kcenter} \\\hline
Asymmetric $k$-center under $\alpha$-PR & $\OPT$ & $2$ & Yes & Yes & Theorem \ref{thm:asy_local} \\\hline
Symmetric $k$-center under $(\alpha,\epsilon)$-PR & $\OPT$ & $3$ & No & Yes & Theorem \ref{thm:3epsthm}      \\\hline
Asymmetric $k$-center under $(\alpha,\epsilon)$-PR & $\epsilon$-close & $3$ & No & No & Theorem \ref{thm:3eps_asy} \\\hline
\end{tabular}
\end{center}
\caption{Our results over all variants of $k$-center under perturbation resilience}\label{tab:results}
\end{table}

\subsection{Related work} \label{sec:related_work}

\paragraph{Clustering.}
There are three classic 2-approximations for $k$-center from the 1980's
\citep{gonzalez1985clustering,hochbaum1985best,dyer1985simple},
which are known to be tight \citep{hochbaum1985best}.
Asymmetric $k$-center proved to be a much harder problem.
The first nontrivial result was an $O(\log^* n)$ approximation algorithm \citep{vishwanathan},
and this was later improved to $O(\log^* k)$ \citep{archer2001two}.
This result was later proven to be asymptotically tight~\citep{chuzhoy2005asymmetric}.

The first constant-factor approximation algorithm for $k$-median was given by
\cite{charikar1999constant},
and the current best approximation ratio is 2.675 by \cite{byrka2015improved}.
\cite{jain2002new} proved $k$-median is NP-hard to approximate to a factor better than 1.73. 

\paragraph{Perturbation resilience.}
Perturbation resilience was introduced by \cite{bilu2012stable}, who showed algorithms that outputted
the optimal solution for max cut under $\Omega(\sqrt{n})$-perturbation resilience (this was later improved by \cite{makarychev2014bilu}).
The study of clustering under perturbation resilience was initiated by \cite{awasthi2012center},
who provided an optimal algorithm for center-based clustering objectives (which includes $k$-median, $k$-means, and $k$-center clustering,
as well as other objectives) under 3-perturbation resilience.
This result was improved by \cite{balcan2012clustering}, who showed an algorithm for center-based clustering under $(1+\sqrt{2})$-perturbation resilience. 
They also gave a near-optimal algorithm for $k$-median under $(2+\sqrt{3},\epsilon)$-perturbation resilience when 
the optimal clusters are not too small.

Recently,
\cite{angelidakis2017algorithms} gave algorithms for center-based clustering (including $k$-median, $k$-means, and $k$-center) under 2-perturbation resilience, 
and defined the more general notion of metric perturbation resilience, although their algorithm does not extend to the $(\alpha,\epsilon)$-perturbation resilience
or local perturbation resilience settings.
\cite{cohen2017one} showed that local search outputs the optimal $k$-median, $k$-means, and $k$-center solution when the data satisfies a stronger
variant of 3-perturbation resilience, in which both the optimal clustering and optimal centers are not allowed to change under any 3-perturbation.
Perturbation resilience has also been applied to other problems, such as min multiway cut, the traveling salesman problem, finding Nash 
equilibria, metric labeling, and facility location~\citep{makarychev2014bilu,mihalak2011complexity,balcan2017nash,lang2017alpha,manthey2018perturbation}.

\paragraph{Subsequent work.}
\cite{vijayaraghavan2017clustering} study $k$-means under additive perturbation resilience, in which the optimal solution cannot change under additive perturbations
to the input distances.
\cite{deshpande2018clustering} gave an algorithm for Euclidean $k$-means under perturbation resilience which runs in time linear in $n$ and the dimension $d$,
and exponentially in $k$ and $\frac{1}{\alpha-1}$.
\cite{chekuri2018perturbation} showed the natural LP relaxation of $k$-center and asymmetric $k$-center is integral for 2-perturbation resilient instances.
They also define a new model of perturbation resilience for clustering with outliers, and they show the algorithm of \cite{angelidakis2017algorithms} exactly solves
clustering with outliers under 2-perturbation resilience, and they further show the natural LP relaxation for $k$-center with outliers is integral for
2-perturbation resilient instances.
Their algorithms have the desirable property that either they output the optimal solution, or they guarantee the input did not satisfy 2-perturbation resilience
(but note this is not the same thing as determining whether or not a given instance satisfies perturbation resilience).

\paragraph{Other stability notions.}
A related notion, approximation stability \citep{as}, states that any $(1+\alpha)$-approximation to the objective
must be $\epsilon$-close to the target clustering. There are several positive results for
$k$-means, $k$-median \citep{as,balcan2009agnostic,gupta2014decompositions}, and min-sum \citep{as,balcan2009finding,voevodski2011min}
under approximation stability.
\cite{ostrovsky2012effectiveness} show how to efficiently cluster instances
in which the $k$-means clustering cost is much lower than the $(k-1)$-means cost.
\cite{kumar2010clustering} give an efficient clustering algorithm for instances in which
the projection of any point onto the line between its cluster center to any other cluster center
is a large additive factor closer to its own center than the other center.
This result was later improved along multiple axes by \cite{Awasthi2012Improved}.
There are many other works that show positive results for different natural notions of stability in various 
settings~\citep{arora2012,awasthi2010stability,gupta2014decompositions,hardt2013beyond,kumar2010clustering,kumar2004simple,tim}.

\section{Preliminaries and basic properties}\label{sec:prelim}


A clustering instance $(S,d)$ consists of a set $S$ of $n$ points, a distance function
$d:S\times S\rightarrow\mathbb{R}_{\geq 0}$, and an integer $k$.
For a point $u\in S$ and a set $A\subseteq S$, we define $d(A,u)=\min_{v\in A}d(v,u)$.
The $k$-center objective is to find a set of points
$X= \{x_1, \dots, x_k\}\subseteq S$ called \emph{centers}
to minimize $\max_{v\in S}d(X,v)$.  
We denote $\text{Vor}_{X,d}(x)=\{v\in S\mid x=\text{argmin}_{y\in X}d(y,v)\}$, 
the Voronoi tile of $x\in X$ induced by $X$ on the set of points $S$ in metric $d$,
and we denote $\text{Vor}_{X,d}(X')=\bigcup_{x\in X'}\text{Vor}_X(x)$ 
for a subset $X'\subseteq X$. We often write $\text{Vor}_{X}(x)$ and $\text{Vor}_{X}(X')$ when $d$ is clear from context.
We refer to the Voronoi partition induced by $X$ as a clustering.
Throughout the paper, we denote the clustering with minimum cost by $\OPT=\{C_1, \dots, C_k \}$, 
we denote the radius of $\OPT$ by $r^*$,
and we denote the optimal centers by $c_1,\dots,c_k$, where $c_i$ is the center of $C_i$ for all $1\leq i\leq k$.
We use $B_r(c)$ to denote a ball of radius $r$ centered at point $r$.

Some of our results assume distance functions which are metrics, and some of our results assume \emph{asymmetric} distance functions.
A distance function $d$ is a \emph{metric} if 
\begin{enumerate}
\item for all $u,v$, $d(u,v)\geq 0$,
\item for all $u,v$, $d(u,v)=0$ if and only if $u=v$, 
\item for all $u,v,w$, $d(u,w)\leq d(u,v)+d(v,w)$, and 
\item for all $u,v$, $d(u,v)=d(v,u)$.
\end{enumerate}
An \emph{asymmetric} distance function satisifies \emph{(1)}, \emph{(2)}, and \emph{(3)}, but not \emph{(4)}.

Now we formally define \emph{perturbation resilience}, a notion introduced by \cite{bilu2012stable}.
$d'$ is called an $\alpha$-perturbation of the distance function $d$, if for all $u,v\in S$, $d(u,v)\leq d'(u,v) \leq \alpha d(u,v)$.
\footnote{
We only consider perturbations in which the distances increase because without loss of generality we can scale
the distances to simulate decreasing distances.
}

\begin{definition}(Perturbation resilience) \label{def:pr}
A clustering instance $(S,d)$ satisfies \emph{$\alpha$-perturbation resilience} ($\alpha$-PR) 
if for any $\alpha$-perturbation $d'$ of $d$, the optimal clustering
$\mathcal{C'}$ under $d'$ is unique and equal to $\OPT$.
\end{definition}

Note that the optimal \emph{centers} might change under an $\alpha$-perturbation,
but the optimal \emph{clustering} must stay the same.
We also consider a relaxed variant of $\alpha$-perturbation resilience, called $(\alpha,\epsilon)$-perturbation resilience,
that allows a small change in the optimal clustering when distances are perturbed. We say that two 
clusterings $\mathcal{C}$ and $\mathcal{C}'$ are $\epsilon$-close if
$\min_\sigma\sum_{i=1}^k \left|C_i\setminus C_{\sigma(i)}'\right|\leq\epsilon n$, where $\sigma$ is a permutation on $[k]$.

\begin{definition}($(\alpha,\epsilon)$-perturbation resilience) \label{def:alpha-epsilon}
A clustering instance $(S,d)$ satisfies $(\alpha,\epsilon)$-perturbation resilience if for any $\alpha$-perturbation
$d'$ of $d$, any optimal clustering $\mathcal{C}'$ under $d'$ is $\epsilon$-close to $\OPT$.
\end{definition}

In Definitions \ref{def:pr} and \ref{def:alpha-epsilon}, we do not assume that the $\alpha$-perturbations satisfy the triangle inequality.
\cite{angelidakis2017algorithms} recently studied the weaker definition in which the $\alpha$-perturbations 
must satisfy the triangle inequality, called \emph{metric perturbation resilience}.
We can update these definitions accordingly. For symmetric clustering objectives, $\alpha$-metric perturbations are restricted to metrics.
For asymmetric clustering objectives, the $\alpha$-metric perturbations must satisfy the directed triangle inequality.

\begin{definition}(Metric perturbation resilience) \label{def:mpr}
A clustering instance $(S,d)$ satisfies \emph{$\alpha$-metric perturbation resilience} ($\alpha$-MPR) 
if for any $\alpha$-metric perturbation $d'$ of $d$, the optimal clustering $\mathcal{C'}$ under $d'$ is unique and equal to $\OPT$.
\end{definition}

In our arguments, we will sometimes convert a non-metric perturbation $d'$ into a metric perturbation by taking the
\emph{metric completion} $d''$ of $d'$ (also referred to as the \emph{shortest-path metric} on $d'$)
by setting the distances in $d''$ as the length of the shortest path on the graph whose edges are the lengths in $d'$.
Note that for all $u,v$, we have $d(u,v)\leq d''(u,v)$ since $d$ was originally a metric.

\subsection{Local Perturbation Resilience}

Now we define perturbation resilience for an optimal cluster rather than the entire dataset.
All prior work has considered perturbation resilience with respect to the entire dataset.


\begin{definition}(Local perturbation resilience)  
Given a clustering instance $(S,d)$ with optimal clustering $\mathcal{C}=\{C_1,\dots,C_k\}$,
an optimal cluster $C_i$ satisfies $\alpha$-perturbation resilience ($\alpha$-PR) if
for any $\alpha$-perturbation $d'$ of $d$, the optimal clustering $\mathcal{C}'$ under
$d'$ is unique and contains $C_i$.
\end{definition}

As a sanity check, we show that a clustering is perturbation resilient if and only if every
optimal cluster satisfies perturbation resilience.

\begin{fact} \label{fact:local-iff}
A clustering instance $(S,d)$ satisfies $\alpha$-PR if and only if each optimal cluster satisfies $\alpha$-PR.
\end{fact}

\begin{proof}
Given a clustering instance $(S,d)$, the forward direction follows by definition:
assume $(S,d)$ satisfies $\alpha$-PR, and given an optimal cluster $C_i$, then for each $\alpha$-perturbation $d'$, the optimal clustering
stays the same under $d'$, therefore $C_i$ is contained in the optimal clustering under $d'$.
Now we prove the reverse direction. Given a clustering instance with optimal clustering $\mathcal{C}$, and given an $\alpha$-perturbation
$d'$, let the optimal clustering under $d'$ be $\mathcal{C}'$. For each $C_i\in\mathcal{C}$, by assumption, $C_i$ satisfies $\alpha$-PR,
so $C_i\in\mathcal{C}'$. Therefore $\mathcal{C}=\mathcal{C}'$.
\end{proof}

Next we define the local version of $(\alpha,\epsilon)$-PR.

\begin{definition}(Local $(\alpha,\epsilon)$-perturbation resilience) \label{def:local-alpha-epsilon}
Given a clustering instance $(S,d)$ with optimal clustering $\mathcal{C}=\{C_1,\dots,C_k\}$, 
an optimal cluster $C_i$ satisfies $(\alpha,\epsilon)$-PR if for any $\alpha$-perturbation $d'$ of $d$,
the optimal clustering $\mathcal{C}'$ under $d'$ contains a cluster $C_i'$ which is $\epsilon$-close to $C_i$.
\end{definition}

In Sections \ref{sec:local} and \ref{sec:3eps}, we will consider a slightly stronger notion of local perturbation resilience.
Informally, an optimal cluster satisfies $\alpha$-strong local perturbation resilience if it is $\alpha$-PR, and all nearby optimal 
clusters are also $\alpha$-PR.
We will sometimes be able to prove guarantees for clusters satisfying strong local perturbation resilience which are not true under standard local perturbation resilience.

\begin{definition}(Strong local perturbation resilience)  \label{def:spr}
Given a clustering instance $(S,d)$ with optimal clustering $\mathcal{C}=\{C_1,\dots,C_k\}$,
an optimal cluster $C_i$ satisfies $\alpha$-strong local perturbation resilience ($\alpha$-SLPR) if
for each $j$ such that there exists $u\in C_i$, $v\in C_j$, and $d(u,v)\leq r^*$, then
$C_j$ is $\alpha$-PR (any cluster that is close to $C_i$ must be $\alpha$-PR).
\end{definition}

To conclude this section, we state a lemma for asymmetric (and symmetric) $k$-center
which allows us to reason about a specific class of $\alpha$-perturbations 
which will be important throughout the paper.
We give two versions of the lemma, each of which will be useful in different sections of the paper.

\begin{lemma} \label{lem:d'_same_cost}
Given a clustering instance $(S,d)$ and $\alpha\geq 1$,
\begin{enumerate}
\item assume we have an $\alpha$-perturbation $d'$ of $d$ with the following property:
for all $p,q$, if $d(p,q)\geq r^*$ then $d'(p,q)\geq \alpha r^*$.
Then the optimal cost under $d'$ is $\alpha r^*$.

\item assume we have an $\alpha$-perturbation $d'$ of $d$ with the following property:
for all $u,v$, either $d'(u,v)=\min(\alpha r^*,\alpha d(u,v))$ or $d'(u,v)=\alpha d(u,v)$.
Then the optimal cost under $d'$ is $\alpha r^*$.
\end{enumerate}
\end{lemma}

\begin{proof} 
\begin{enumerate}
\item Assume there exists a set of centers $C'=\{c_1',\dots, c_k'\}$ whose $k$-center cost under $d'$
is $<\alpha r^*$. Then for all $i$ and $s\in\text{Vor}_{C',d'}(c_i')$, $d'(c_i',s)<\alpha r^*$,
implying $d(c_i',s)< r^*$ by construction. It follows that the $k$-center cost of
$C'$ under $d$ is $r^*$, which is a contradiction.
Therefore, the optimal cost under $d'$ must be $\alpha r^*$.

\item Given $u,v$ such that $d(u,v)\geq r^*$, then $d'(u,v)\geq \alpha r^*$ by construction.
Now the proof follows from part one.
\end{enumerate}
\end{proof}

\section{$k$-center under perturbation resilience} \label{sec:2pr}

In this section, 
we provide efficient algorithms for finding the optimal clustering for symmetric and asymmetric instances of 
$k$-center under $2$-perturbation resilience.
Our results directly improve on the result by \cite{balcan2012clustering} for symmetric $k$-center under $(1+\sqrt{2})$-perturbation resilience. 
We also show that it is NP-hard to recover $\OPT$ even for symmetric $k$-center instance under $(2-\delta)$-perturbation resilience. 
As an immediate consequence, our results are tight for both symmetric and asymmetric $k$-center instances.
This is the first problem for which the exact value of perturbation resilience
is found ($\alpha=2$), where the problem switches from efficiently computable to NP-hard.

First, we show that any $\alpha$-approximation algorithm
returns the optimal solution for $\alpha$-perturbation resilient instances.
An immediate consequence is an algorithm for symmetric $k$-center under 2-perturbation resilience.
Next, we provide a novel algorithm for asymmetric $k$-center under 2-perturbation resilience.
Finally, we show hardness of $k$-center under $(2-\delta)$-PR.

\subsection{$\alpha$-approximations are optimal under $\alpha$-PR}

The following theorem shows that any $\alpha$-approximation
algorithm for $k$-center will return the optimal solution on clustering
instances that are $\alpha$-perturbation resilient.

\begin{theorem} \label{thm:approximation}
Given a clustering instance $(S,d)$ satisfying $\alpha$-perturbation resilience for 
asymmetric $k$-center, and a set $C$ of $k$ centers which is an $\alpha$-approximation, i.e., $\forall p\in S$, 
$\exists c\in C$ such that $d(c,p)\leq \alpha r^*$,
then the Voronoi partition induced by $C$ is the optimal clustering.
\end{theorem}

\begin{proof}
For a point $p\in S$, let $c(p):=\text{argmin}_{c\in C} d(c,p)$, 
the closest center in $C$ to $p$. 
The idea is to construct an $\alpha$-perturbation in which $C$ is the
optimal solution by increasing all distances except between $p$
and $c(p)$, for all $p$. Then the theorem will follow by using the
definition of perturbation resilience.

By assumption, $\forall p\in S$, $d(c(p),p)\leq \alpha r^*$.
Create a perturbation $d'$ as follows. 
Increase all distances  by a factor of $\alpha$, 
except for all $p\in S$,
set $d'(c(p),p)=\min(\alpha d(c(p),p),\alpha r^*)$ (recall in Definition \ref{def:pr}, the
perturbation need not satisfy the triangle inequality).
Then no distances were increased by more than a
factor of $\alpha$. And since we had that $d(c(p),p)\leq \alpha r^*$, no distances decrease either.
Therefore, $d'$ is an $\alpha$-perturbation of $d$.
By Lemma \ref{lem:d'_same_cost}, the optimal cost for $d'$ is $\alpha r^*$.
Also, $C$ achieves cost $\leq\alpha r^*$ by construction, 
so $C$ is an optimal set of centers
under $d'$. Then by $\alpha$-perturbation resilience, 
the Voronoi partition induced by $C$ under $d'$ is
the optimal clustering. 

Finally, we show the Voronoi partition of $C$ under $d$ is the same as the
Voronoi partition of $C$ under $d'$. Given $p\in S$ whose closest point in $C$
is $c(p)$ under $d$, then under $d'$, all distances from $p$ to $C\setminus\{c(p)\}$ increased by exactly $\alpha$, and $d(p,c(p))$ increased by $\leq\alpha$.
Therefore, the closest point in $C$ to $p$ under $d'$ is still $c(p)$.
\end{proof}

\subsection{$k$-center under 2-PR} 

An immediate consequence of Theorem \ref{thm:approximation} is that 
we have an exact algorithm for symmetric $k$-center under 2-perturbation resilience
by running a simple 2-approximation algorithm (e.g., \citep{gonzalez1985clustering,hochbaum1985best,dyer1985simple}).
However, Theorem \ref{thm:approximation} only gives an algorithm for asymmetric $k$-center
under $O(\log^*(k))$-perturbation resilience.
Next, we show it is possible to substantially improve the latter result.

\subsubsection{Asymmetric $k$-center under 2-PR} \label{sec:2pr-asymmetric}
One of the challenges involved in dealing with asymmetric $k$-center instances 
is the fact that even though for all $p\in C_i$,  
$d(c_i,p)\leq r^*$, the reverse distance, $d(p,c_i)$, might be arbitrarily large.
Such points for which $d(p,c_i)\gg r^*$ pose a challenge to the structure of the clusters, as
they can be very close to points or even centers of other clusters.
To deal with this challenge,
we first define a set of ``good'' points, $A$, such that
$A=\{p\mid \forall q, d(q,p)\leq r^* \implies d(p,q)\leq r^*\}$.
\footnote{
A ``good'' point is referred to as a \emph{center-capturing vertex} in other works, e.g., \citep{vishwanathan}.
We formally define this notion in Section \ref{sec:local}.
}
Intuitively speaking, these points behave similarly to a set of points with symmetric distances up to a distance $r^*$. To explore this, we define a desirable property of $A$ with respect to the optimal clustering.
\vspace*{-.2cm}
\begin{definition} \label{def:respect}
$A$ is said to \emph{respect the structure of $\OPT$} if  

(1) $c_i\in A$ for all $i\in [k]$, and 

(2) for all $p\in S\setminus A$, if $A(p):=\arg\min_{q\in A}d(q,p)\in C_i$, then $p\in C_i$.
\end{definition}
\vspace*{-.2cm}
For all $i$, define $C_i'=C_i\cap A$ (which is in fact the optimal clustering of $A$).
Satisfying Definition \ref{def:respect} implies that if we can optimally cluster $A$,
then we can optimally cluster the entire
instance (formalized in Theorem \ref{thm:akc2pr}).
Thus our goal is to show that $A$ does indeed respect the structure of $\OPT$, and to show how
to return $C_1',\dots, C_k'$.

Intuitively, $A$ is similar to a symmetric 2-perturbation resilient clustering instance.
However, some structure is no longer there, for instance, a point $p$ may be at distance
$\leq 2r^*$ from every point in a different cluster, which is not true for 2-perturbation resilient
instances. 
This implies we cannot simply run a 2-approximation algorithm on the set $A$,
as we did in the previous section.
However, we show that the remaining structural properties are sufficient to 
optimally cluster $A$.
To this end, we define two properties and show how they lead to an algorithm that returns $C_1',\dots,C_k'$,
and help us prove that $A$ respects the structure of $\OPT$.

The first of these properties requires each point to be closer to its center than any point in another cluster.

\emph{Property (1): For all $p\in C_i'$ and $q\in C_{j\neq i}'$, $d(c_i,p)<d(q,p)$. }

The second property requires that any point within distance $r^*$ of a cluster center belongs to that cluster.

\emph{Property (2): For all $i\neq j$ and $q\in C_j$,  $d(q,c_i) > r^*$}
(see Figure \ref{fig:2pr-mainbody}).
\footnote{
Property (1) first appeared  in the work of \cite{awasthi2012center}, 
for symmetric clustering instances.
A weaker variation of Property (2) was introduced by \cite{balcan2012clustering}, which showed that in $1+\sqrt 2$-perturbation resilient instances for any cluster $C_i$ with radius $r_i$, $B_{r_i}(c_i)=C_i$. Our Property (2) shows that this is true \emph{for a universal radius, $r^*$}, even for $2$-perturbation resilient instances, \emph{and}
even for asymmetric instances.
}

\begin{figure}
    \centering
    \begin{subfigure}[b]{0.35\textwidth}
        \includegraphics[width=\textwidth]{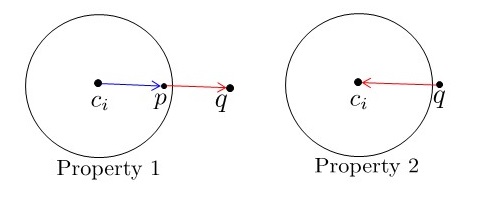}
        \caption{\small Properties of $2$-perturbation resilience}
        \label{fig:properties}
    \end{subfigure}
    \quad\quad\quad\quad
    \begin{subfigure}[b]{0.45\textwidth}
        \includegraphics[width=\textwidth]{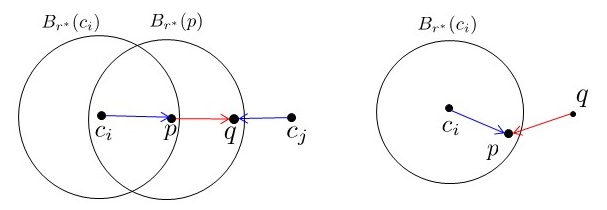}
        \caption{\small Demonstrating the correctness of Algorithm~\ref{alg:asymmetric}}
        \label{fig:proof-alg}
    \end{subfigure}
    \caption{Properties of a $2$-perturbation resilient instance of aymmetric $k$-center that are used for clustering. }
    \label{fig:2pr-mainbody}
\end{figure}

Let us illustrate how these properties allow us to optimally cluster $A$. 
\footnote{
Other algorithms work, such as single linkage with dynamic programming at the end
to find the minimum cost pruning of $k$ clusters. However, our algorithm is able
to recognize optimal clusters \emph{locally} (without a complete view of the point set).
}
Consider a ball of radius $r^*$ around a center $c_i$. By Property 2, 
such a ball exactly captures $C_i'$. Furthermore, by Property 1, any point in this ball is closer to the center than to points outside of the ball. Is this true for a ball of radius $r^*$ around a general point $p$? Not necessarily. 
If this ball contains a point $q\in C_j'$ from a different cluster, then $q$ will be closer to a point outside
the ball than to $p$ (namely, $c_j$, which is guaranteed to be outside of the ball by Property 2).
This allows us to determine that the center of such a ball must not be an optimal center.

This structure motivates our Algorithm \ref{alg:asymmetric} for asymmetric $k$-center under
2-perturbation resilience.
At a high level, we start by constructing the set $A$ (which can be done
easily in polynomial time).
Then we create the
set of all balls of radius $r^*$ around all points in $A$ 
(if $r^*$ is not known, we can use a guess-and-check wrapper). 
Next, we prune this set by throwing out any ball that contains a point farther from its center than to a point outside the ball. We also throw out any ball that is a subset of another one. Our claim is that the remaining balls are exactly $C_1',\dots,C_k'$.
Finally, we add the points in $S\setminus A$ to their closest point in $A$.

\begin{algorithm}[h]
\caption {\textsc{Asymmetric $k$-center algorithm under 2-PR}}
\label{alg:asymmetric}
\begin{algorithmic}
\Require{ Asymmetric $k$-center instance $(S, d)$, distance $r^*$ (or try all possible candidates)} \\
\textbf{Create symmetric set}
\begin{itemize}
\item Build set $A=\{p\mid \forall q, d(q,p)\leq r^* \implies d(p,q)\leq r^*\}$ 
\end{itemize}
\textbf{Create candidate balls}
\begin{itemize}
\item $\forall c\in A$, construct $G_c=B_{r^*}(c)$
(the ball of radius $r^*$ around $c$). 
\end{itemize}
\textbf{Prune balls}
\begin{itemize}
\item $\forall G_c$, if $\exists p\in G_c$, $q\notin G_c$ such that $d(q,p)<d(c,p)$, then
throw out $G_c$. \label{step2} 
\item $\forall p,q$ such that $G_p\subseteq G_q$, throw out $G_p$. \label{step3} 
\end{itemize}
\textbf{Insert remaining points}
\begin{itemize}
\item $\forall p\notin A$, add $p$ to $G_q$, where $q=\arg\min_{s\in A}d(s,p)$.
\end{itemize}
\Ensure{sets $G_1,\dots,G_k$ }
\end{algorithmic}
\end{algorithm}

\paragraph{Formal details of our analysis}

\begin{lemma}\label{lem:properties}
Properties 1 and 2 hold for asymmetric $k$-center instances satisfying 2-perturbation resilience.
\end{lemma}

\begin{proof}

\emph{Property 1:}
Assume false, $d(q,p)\leq d(c_i,p)$.
The idea will be that since $q$ is in $A$, it is close to its own center, so we can
construct a perturbation in which $q$ replaces its center $c_j$.
Then $p$ will join $q$'s cluster, causing a contradiction.
Construct the following $d'$:

\begin{equation*}
d'(s,t)=
\begin{cases}
\min( 2 r^*, 2 d(s,t) ) & \text{if } s=q \text{, } t\in C_j\cup\{p\}  \\
2 d(s,t) & \text{otherwise.}
\end{cases}
\end{equation*}

This is a $2$-perturbation because $d(q,C_j\cup\{p\})\leq 2 r^*$.
Then by Lemma~\ref{lem:d'_same_cost}, the optimal cost is $2 r^*$.
The set of centers $\{c_\ell\}_{i=1}^k\setminus\{c_j\}\cup\{q\}$ achieves the optimal
cost, since $q$ is distance $2 r^*$ from $C_j$, and all other 
clusters have the same center as in $\mathcal{OPT}$ (achieving radius $2 r^*$).
Then for all $c_\ell$, $d'(q,p)\leq d'(c_i,p)\leq d'(c_\ell,p)$.
And since $q\in A$, $d(q,c_j)\leq r^*$ so $d(q,C_j)\leq 2r^*$.
Then we can construct a 2-perturbation in which $q$ becomes the center of $C_j$,
and then $q$ is the best center for $p$, so we have a contradiction.

\emph{Property 2:}
Assume on the contrary that there exists $q\in C_j$, $i\neq j$ such that $d(q,c_i)\leq r^*$. 
Now we will define a $d'$ in which $q$ can become a center for $C_i$.

\begin{equation*}
d'(s,t)=
\begin{cases}
\min( 2 r^*,2 d(s,t) ) & \text{if } s=q \text{, } t\in C_i  \\
2 d(s,t) & \text{otherwise.}
\end{cases}
\end{equation*}

This is a 2-perturbation because $d(q,C_i)\leq 2 r^*$.
Then by Lemma~\ref{lem:d'_same_cost}, the optimal cost is $2 r^*$.
The set of centers $\{c_\ell\}_{i=1}^k\setminus\{c_i\}\cup\{q\}$ achieves the optimal
cost, since $q$ is distance $2 r^*$ from $C_i$, and all other 
clusters have the same center as in $\mathcal{OPT}$ (achieving radius $2 r^*$).
But the clustering with centers $\{c_\ell\}_{i=1}^k\setminus\{c_i\}\cup\{q\}$ is different
from $\mathcal{OPT}$, since (at the very least) $q$ and $c_i$ are in different clusters.
This contradicts 2-perturbation resilience.
\end{proof}

\begin{lemma}\label{lem:respect}
The set $A$ respects the structure of $\OPT$.
\end{lemma}

\begin{proof} 

From Lemma \ref{lem:properties}, we can use Property 2 in our analysis.
First we show that $c_i\in A$ for all $i\in [k]$.
Given $c_i$, $\forall p\in C_i$, then $d(c_i,p)\leq r^*$ by definition of $\OPT$.
$\forall q\notin C_i$, then by Property 2,
$d(q,c_i)>r^*$.
It follows that for any point $p\in S$, it cannot be the case that
$d(p,c_i)\leq r^*$ and $d(c_i,p)>r^*$. Therefore, $c_i\in A$.

Now we show that 
for all $p\in S\setminus A$, if $A(p)\in C_i$, then $p\in C_i$.
Given $p\in S\setminus A$, let $p\in C_i$ and assume towards contradiction that $q=A(p)\in C_j$ for some $i\neq j$.
We will construct a 2-perturbation $d'$ in which $q$ replaces $c_j$ as the center for $C_j$ and
$p$ switches from $C_i$ to $C_j$, causing a contradiction.
We construct $d'$ as follows.
All distances are increased by a factor of $2$ except for $d(q,p)$ and $d(q,q')$ for all $q'\in C_j$.
These distances are increased by a factor of $2$ up to $2 r^*$.
Formally,

\begin{equation*}
d'(s,t)=
\begin{cases}
\min( 2 r^*, 2 d(s,t) ) & \text{if } s=q \text{, } t\in C_j\cup\{p\}  \\
2 d(s,t) & \text{otherwise.}
\end{cases}
\end{equation*}

This is a $2$-perturbation because $d(q,C_j)\leq 2 r^*$.
Then by Lemma~\ref{lem:d'_same_cost}, the optimal cost is $2 r^*$.
The set of centers $\{c_\ell\}_{i=1}^k\setminus\{c_j\}\cup\{q\}$ achieves the optimal
cost, since $q$ is distance $2 r^*$ from $C_j$, and all other 
clusters have the same center as in $\mathcal{OPT}$ (achieving radius $2 r^*$).
But consider the point $p$. Since all centers are in $A$ and $q$ is the closest point
to $p$ in $A$, then $q$ is the center for $p$ under $d'$.
Therefore, the optimal clustering under $d'$ is different from $\OPT$, so we have
a contradiction.
\end{proof}

Now we are ready to show Algorithm \ref{alg:asymmetric} returns the optimal clustering.

\begin{theorem} \label{thm:akc2pr}
Algorithm \ref{alg:asymmetric} returns the exact solution for asymmetric $k$-center
under 2-perturbation resilience.
\end{theorem}

\begin{proof}
In this proof, we refer to the first line of \textbf{Prune balls} in Algorithm~\ref{alg:asymmetric} as \textbf{Pruning step 1}
and the second line as \textbf{Pruning step 2}.
First we must show that after \textbf{Pruning step 2}, the remaining sets are exactly
$C_1',\dots,C_k'=C_1\cap A,\dots,C_k\cap A$.
We prove this in three steps: the sets $G_{c_i}$
correspond to $C_i'$, these sets are not thrown out in \textbf{Pruning step 1} and \textbf{Pruning step 2}, 
and all other sets are thrown out in steps \textbf{Pruning step 1} and \textbf{Pruning step 2}.
Because of Lemma \ref{lem:properties}, we can use Properties 1 and 2.

For all $ i$, $G_{c_i}=C_i'$:
From Lemma \ref{lem:respect}, all centers are in $A$, so $G_{c_i}$ will be created in
step 2.
For all $p\in C_i$, $d(c_i,p)\leq r^*$.
For all $q\notin C_i'$, then by Property 2, $d(q,c_i)>r^*$
(and since $c_i,q\in A$, $d(c_i,q)>r^*$ as well).
For all $i$, $G_{c_i}$ is not thrown out in step \textbf{Pruning step 1}:
Given $s\in G_{c_i}$ and $t\notin G_{c_i}$. Then $s\in C_i'$ and $t\in C_j'$ 
for $j\neq i$.
If $d(t,s)<d(c_i,s)$, then we get a contradiction from Property 1.
For all non-centers $p$, $G_p$ is thrown out in \textbf{Pruning step 1} or \textbf{Pruning step 2}:
From the previous paragraph, $G_{c_i}=C_i'$.
If $G_p\subseteq G_{c_i}$, then $G_p$ will be thrown out in \textbf{Pruning step 2}
(if $G_p=G_{c_i}$, it does not matter which set we keep, 
so without loss of generality say that we keep $G_{c_i}$).
Then if $G_p$ is not thrown out in \textbf{Pruning step 2}, $\exists s\in G_p\cap C_j'$, $j\neq i$.
If $s=c_j$, then $d(p,c_j)\leq r^*$ and we get a contradiction from Property 2.
So, we can assume $s$ is a non-center (and that $c_j\notin G_p$).
But $d(c_j,s)<d(p,s)$ from Property 1, and therefore $G_p$ will be thrown
out in \textbf{Pruning step 1}.
Thus, the remaining sets after \textbf{Pruning step 2} are exactly $C_1',\dots,C_k'$.

Finally, by Lemma \ref{lem:respect}, for each $p\in C_i\setminus A$,
$A(p)\in C_i$, so $p$ will be added to $G_{c_i}$. 
Therefore, the final output is $C_1,\dots,C_k$.
\end{proof}

\subsubsection{Hardness of $k$-center under perturbation resilience} \label{sec:hardness}

In this section, we show NP-hardness for $k$-center under $(2-\delta)$-perturbation resilience.
We show that if there exists a polynomial time algorithm which returns the optimal solution for
symmetric $k$-center under $(2-\delta)$-perturbation resilience,
\footnote{
In fact, our result holds even under the strictly stronger notion of 
\emph{approximation stability}~\citep{as}.
}
then $NP=RP$, 
even under the condition that the optimal clusters are all size $\geq\frac{n}{2k}$.
Because symmetric $k$-center is a special case of asymmetric
$k$-center, we have the same  hardness results for asymmetric $k$-center.
This proves Theorem \ref{thm:akc2pr} is tight with respect to the level of perturbation resilience assumed.

\begin{theorem} \label{thm:2eps_as_hard}
There is no polynomial time algorithm for finding
the optimal $k$-center clustering under $(2-\delta)$-perturbation resilience, even when assuming all optimal clusters are size $\geq \frac{n}{2k}$, unless $NP=RP$.
\end{theorem}

We show a reduction from a special case of Dominating Set which we call Unambiguous-Balanced-Perfect Dominating Set.
Below, we formally define this problem and all intermediate problems.
Part of our reduction is based off of the proof of \cite{lev}, who showed a reduction from a variant of dominating set
to the weaker problem of clustering under $(2-\delta)$-center proximity.
\footnote{
$\alpha$-center proximity is the property that for all $p\in C_i$ and $j\neq i$, $\alpha d(c_i,p)<d(c_j,p)$, and it follows from $\alpha$-perturbation resilience.
}
We use four NP-hard problems in a chain of reductions. 
Here, we define all of these problems up front.
We introduce the ``balanced'' variants of two existing problems.

\begin{definition}[3-Dimensional Matching (3DM) \citep{karp1972reducibility}]
We are given three disjoint sets $X_1$, $X_2$, and $X_3$ each of size $m$, and a set $T$ such that $t\in T$
is a triple $t=(x_1,x_2,x_3)$ where $x_1\in X_1$, $x_2\in X_2$, and $x_3\in X_3$.
The problem is to find a set $M\subseteq T$ of size $m$ which exactly hits all the elements
in $X_1\cup X_2\cup X_3$. In other words, for all pairs $(x_1,x_2,x_3),(y_1,y_2,y_3)\in M$, it is the case that
$x_1\neq y_1$, $x_2\neq y_2$, and $x_3\neq y_3$.
\end{definition}

\begin{definition}[Balanced-3-Dimensional Matching (B3DM)]
This is the 3DM problem $(X_1,X_2,X_3,T)$ with the additional constraint that
$2m\leq |T|\leq 3m$, where $|X_1|=|X_2|=|X_3|=m$.
\end{definition}

\begin{definition}[Perfect Dominating Set (PDS)~\citep{lev}]
Given a graph $G=(V,E)$ and an integer $k$,
the problem is to find a set of vertices $D\subseteq V$ of size $k$ such that for all $v\in V\setminus D$,
there exists \emph{exactly one} $d\in D$ such that $(v,d)\in E$ (then we say $v$ ``hits'' $d$).
\end{definition}

\begin{definition}[Balanced-Perfect-Dominating Set (BPDS)]
This is the PDS problem $(G,k)$ with the additional assumption that if the graph has $n$
vertices and a dominating set of size $k$ exists, then each vertex in the dominating set hits at least $\frac{n}{2k}$ vertices.
\end{definition}

Additionally, each problem has an ``Unambiguous'' variant, which is the added constraint
that the problem has at most one solution.
\cite{valiant} showed that Unambiguous-3SAT is hard unless $NP=RP$.
To show the Unambiguous version of another problem is hard, one must establish a
parsimonious polynomial time reduction from Unambiguous-3SAT to that problem.
A parsimonious reduction is one that conserves the number of solutions.
For two problems $A$ and $B$, we denote $A\leq_{par} B$ to mean there is a reduction
from $A$ to $B$ that is parsimonious \emph{and} polynomial time.
Some common reductions involve 1-to-1 mappings which are easy to verify parsimony, but many other common reductions
are not parsimonious.
For instance, the standard reduction from 3SAT to 3DM is not parsimonious \citep{kleinberg2006algorithm},
yet there is a more roundabout series of reductions which all use 1-to-1 mappings, and are therefore easy to verify parsimony.
In order to prove Theorem \ref{thm:2eps_as_hard}, we start with the claim that Unambiguous-BPDS is hard unless $NP=RP$.
We use a parsimonious series of reductions from 3SAT to B3DM to BPDS.
All of these reductions are from prior work, yet we verify parsimony and balancedness.

\begin{lemma} \label{lem:hardness}
There is no polynomial time algorithm for Unambiguous-BPDS unless $NP=RP$.
\end{lemma}

\begin{proof}
We give a parsimonious series of reductions from 3SAT to B3DM to BPDS.
Then it follows from the result by \cite{valiant} that there is no polynomial time algorithm for Unambiguous-BPDS unless $NP=RP$.

To show that B3DM is NP-hard, we use the reduction of \cite{dyer}
who showed that Planar-3DM is NP-hard.
While planarity is not important for the purpose of our problems, their reduction 
from 3SAT has two other nice properties that we crucially use.
First, the reduction is parsimonious, as pointed out by \cite{hunt}.
Second, given their 3DM instance $X_1,X_2,X_3,T$,
each element in $X_1\cup X_2\cup X_3$ appears in either two or three  tuples in $T$.
(\cite{dyer} mention this observation just before their Theorem 2.3.)
From this, it follows that $2m\leq |T|\leq 3m$,
and so their reduction proves that B3DM is $NP$-hard via a parsimonious reduction from 3SAT.

Next, we reduce B3DM to BPDS using a reduction similar to the reduction by \cite{lev}.
Their reduction maps every element in $X_1\cup X_2\cup X_3\cup T$ to a vertex in $V$, and adds
one extra vertex $v$ to $V$.
There is an edge from each element $(x_1,x_2,x_3)\in T$ to the 
corresponding elements $x_1\in X_1$, $x_2\in X_2$, and $x_3\in X_3$.
Furthermore, there is an edge from $v$ to every element in $T$.
\cite{lev} show that if the 3DM instance is a {\sc yes} instance with
matching $M\subseteq T$ then the minimum dominating set is $v\cup M$.
Now we will verify this same reduction can be used to reduce B3DM to BPDS.
If we start with B3DM, our graph has $|X_1|+|X_2|+|X_3|+|T|+1\leq 6m+1$ vertices since $|T|\leq 3m$, so $n\leq 6m+1$.
Also note that in the {\sc yes} instance, the dominating set is size $m+1$ by construction.
Therefore, to verify the reduction to BPDS, we must show that in the {\sc yes} instance, each node in the dominating set hits $\geq\frac{6m+1}{2(m+1)}$ nodes.
Given $t\in M$, $t$ hits 3 nodes in the graph, and $\frac{n}{2(m+1)}\leq \frac{6m+1}{2m+2}\leq 3$.
The final node in the dominating set is $v$, and $v$ hits $|T|-m\geq 2m-m=m$ nodes, and $\frac{6m+1}{2(m+1)}\leq m$ when $m\geq 3$.
Therefore, the resulting instance is a BPDS instance.

Now we have verified that there exists a parsimonious reduction 
3SAT $\leq_{par}$ BPDS, so it follows that there is no polynomial time algorithm for Unambiguous-BPDS unless $NP=RP$.
\end{proof}

Now we can prove Theorem \ref{thm:2eps_as_hard} by giving a reduction from Unambiguous-BPDS to $k$-center clustering under 
$(2-\delta)$-perturbation resilience, where all clusters are size $\geq \frac{n}{2k}$.
We use the same reduction as \cite{lev}, but
we must verify that the resulting instance is $(2-\delta)$-perturbation resilient.
Note that our reduction requires the Unambiguous variant while the reduction of \cite{lev} does not, since we are reducing to a stronger problem.

\begin{proof}[Proof of Theorem \ref{thm:2eps_as_hard}]
From Lemma \ref{lem:hardness}, Unambiguous-BPDS is NP-hard unless $NP=RP$.
Now for all $\delta>0$, we reduce from Unambiguous-BPDS to $k$-center clustering 
and show the resulting instance has all cluster sizes $\geq \frac{n}{2k}$ 
and satisfies $(2-\delta)$-perturbation resilience.

Given an instance of Unambiguous-BPDS, for every $v\in V$, create a point $v\in S$ in the clustering instance. 
For every edge $(u,v)\in E$, let $d(u,v) = 1$, otherwise let $d(u,v) = 2$. 
Since all distances are either $1$ or $2$, the triangle inequality is trivially satisfied. 
Then a $k$-center solution of cost 1 exists if and only if there exists a dominating set of size $k$.

Since each vertex in the dominating set hits at least $\frac{n}{2k}$ vertices, the resulting clusters will be size at least $\frac{n}{2k}+1$.
Additionally, if there exists a dominating set of size $k$, then the corresponding optimal $k$-center
clustering has cost 1. Because this dominating set is perfect and unique, any other clustering has cost 2.
It follows that the $k$-center instance is $(2-\delta)$-perturbation resilient.
\end{proof}

\section{$k$-center under metric perturbation resilience} \label{sec:metric}

In this section, we extend the results from Section \ref{sec:2pr} to the metric perturbation resilience setting \citep{angelidakis2017algorithms}.
We first give a generalization of Lemma \ref{lem:d'_same_cost} to show that it can be extended to metric perturbation resilience.
Then we show how this immediately leads to corollaries of Theorem \ref{thm:approximation} and Theorem \ref{thm:akc2pr} extended to the
metric perturbation resilience setting.

Recall that in the proofs from the previous section, we created $\alpha$-perturbations $d'$ by increasing all distances by $\alpha$, 
except a few distances $d(u,v)\leq \alpha r^*$ which we increased to $\min(\alpha d(u,v),\alpha r^*)$.
In this specific type of $\alpha$-perturbation, we used the crucial property that the optimal clustering has cost $\alpha r^*$ 
(Lemma \ref{lem:d'_same_cost}).
However, $d'$ may be highly non-metric, so our challenge is arguing that the proof still goes through after taking the metric
completion of $d'$  (recall the metric completion of $d'$ is defined as the shortest path metric on $d'$).
In the following lemma, we show that Lemma \ref{lem:d'_same_cost} remains true after taking the metric completion of the perturbation.

\begin{lemma} \label{lem:d'_metric}
Given $\alpha\geq 1$ and an asymmetric $k$-center clustering instance $(S,d)$ with optimal radius $r^*$,
let $d''$ denote an $\alpha$-perturbation such that for all $u,v$, either
$d''(u,v)=\min(\alpha r^*,\alpha d(u,v))$ or $d''(u,v)=\alpha d(u,v)$.
Let $d'$ denote the metric completion of $d''$.
Then $d'$ is an $\alpha$-metric perturbation of $d$, and the optimal cost under $d'$ is $\alpha r^*$.
\end{lemma}

\begin{proof} 

By construction, $d'(u,v)\leq d''(u,v)\leq\alpha d(u,v)$.
Since $d$ satisfies the triangle inequality, we have that $d(u,v)\leq d'(u,v)$,
so $d'$ is a valid $\alpha$-metric perturbation of $d$.

Now given $u,v$ such that $d(u,v)\geq r^*$, we will prove that $d'(u,v)\geq \alpha r^*$.
By construction, $d''(u,v)\geq\alpha r^*$.
Then since $d'$ is the metric completion of $d''$, there exists a path 
$u=u_0$--$u_1$--$\cdots$--$u_{s-1}$--$u_s=v$
such that $d'(u,v)=\sum_{i=0}^{s-1}d'(u_i,u_{i+1})$ and
for all $0\leq i\leq s-1$, $d'(u_i,u_{i+1})=d''(u_i,u_{i+1})$.

Case 1: there exists an $i$ such that $d''(u_i,u_{i+1})\geq\alpha r^*$.
Then $d'(u,v)\geq\alpha r^*$ and we are done.

Case 2: for all $0\leq i\leq s-1$, $d''(u_i,u_{i+1})<\alpha r^*$.
Then by construction, $d'(u_i,u_{i+1})=d''(u_i,u_{i+1})=\alpha d(u_i,u_{i+1})$,
and so $d'(u,v)=\sum_{i=0}^{s-1}d'(u_i,u_{i+1})=\alpha \sum_{i=0}^{s-1}d(u_i,u_{i+1})\geq
\alpha d(u,v)\geq \alpha r^*$.

We have proven that for all $u,v$, if $d(u,v)\geq r^*$, then $d'(u,v)\geq\alpha r^*$.
Then by Lemma~\ref{lem:d'_same_cost}, the optimal cost under $d'$ must be $\alpha r^*$.
\end{proof}

Recall that metric perturbation resilience states that the optimal solution does not change under
any metric perturbation to the input distances.
In the proofs of Theorems \ref{thm:approximation} and \ref{thm:akc2pr},
the only perturbations constructed were the type as in Lemma \ref{lem:d'_same_cost}.
Since Lemma \ref{lem:d'_metric} shows this type of perturbation is indeed a metric,
Theorems \ref{thm:approximation} and \ref{thm:akc2pr} are true even under metric perturbation resilience.

\begin{corollary} \label{cor:approximation}
Given a clustering instance $(S,d)$ satisfying $\alpha$-metric perturbation resilience for 
asymmetric $k$-center, and a set $C$ of $k$ centers which is an $\alpha$-approximation, i.e., $\forall p\in S$, 
$\exists c\in C$ such that $d(c,p)\leq \alpha r^*$,
then the Voronoi partition induced by $C$ is the optimal clustering.
\end{corollary}

\begin{corollary} \label{cor:akc2pr}
Algorithm \ref{alg:asymmetric} returns the exact solution for asymmetric $k$-center
under 2-metric perturbation resilience.
\end{corollary}

\section{$k$-center under local perturbation resilience} \label{sec:local}

In this section, we further extend the results from Sections \ref{sec:2pr} and \ref{sec:metric} to the local perturbation resilience setting.
First we show that any $\alpha$-approximation to $k$-center will return each optimal $\alpha$-MPR cluster,
i.e., Corollary \ref{cor:approximation} holds even in the local perturbation resilience setting.
Then for asymmetric $k$-center, we show that a natural modification to the
$O(\log^* n)$ approximation algorithm of \cite{vishwanathan}
leads to an algorithm that maintains its performance in the worst case, while exactly returning each optimal cluster located within a 2-MPR region of the dataset.
This generalizes Corollary \ref{cor:akc2pr}.

\subsection{Symmetric $k$-center}

In section \ref{sec:2pr}, we showed that any $\alpha$-approximation algorithm for $k$-center returns the optimal solution
for instances satisfying $\alpha$-perturbation resilience (and this was generalized to metric perturbation resilience in the previous section).
In this section, we extend this result to the local perturbation resilience setting. We show that any $\alpha$-approximation will return
each (local) $\alpha$-MPR cluster. For example, if a clustering instance is half 2-perturbation resilient, running a 2-approximation algorithm will return
the optimal clusters for half the dataset, and a 2-approximation for the other half.

\begin{theorem} \label{thm:kcenter}
Given an asymmetric $k$-center clustering instance $(S,d)$, a set $C$ of $k$ centers which is an $\alpha$-approximation, 
and a clustering $\mathcal{C}$ defined as the Voronio partition induced by $C$,
then each $\alpha$-MPR cluster is contained in $\mathcal{C}$.
\end{theorem}

The proof is very similar to the proof of Theorem \ref{thm:approximation}.
The key difference is that we reason about each perturbation resilient cluster individually, rather than reasoning about the global structure of perturbation
resilience.

\begin{proof}[Proof of Theorem~\ref{thm:kcenter}]

Given an $\alpha$-approximate solution $\mathcal{C}$ to a clustering instance $(S,d)$, and given
an $\alpha$-MPR cluster $C_i$, we will create an $\alpha$-perturbation as follows. 
Define $\mathcal{C}(v):=\text{argmin}_{c\in\mathcal{C}}d(c,v)$.
For all $v\in S$, set $d''(v,\mathcal{C}(v))=\min\{\alpha r^*,\alpha d(v,\mathcal{C}(v))\}$.
For all other points $u\in S$, set $d''(v,u)=\alpha d(v,u)$.
Then by Lemma \ref{lem:d'_metric}, the metric completion $d'$ of $d''$ is an $\alpha$-perturbation of $d$
with optimal cost $\alpha r^*$.
By construction, the cost of $\mathcal{C}$ is $\leq \alpha r^*$ under $d'$, therefore, $\mathcal{C}$ is an optimal clustering.
Denote the set of centers of $\mathcal{C}$ by $C$.
By definition of $\alpha$-MPR, there exists $v_i\in C$ such that
$\text{Vor}_{C,d'}(v_i)=C_i$. Now, given $v\in C_i$, $\text{argmin}_{u\in C}d'(u,v)=v_i$,
so by construction, $\text{argmin}_{u\in C}d(u,v)=v_i$. Therefore, $\text{Vor}_{C,d}(v_i)=C_i$, so $C_i\in\mathcal{C}$.

\end{proof}

\subsection{Asymmetric $k$-center}

In Section \ref{sec:2pr}, we gave an algorithm which outputs the optimal clustering for asymmetric $k$-center under 2-perturbation resilience 
(Algorithm \ref{alg:asymmetric} and Theorem \ref{thm:akc2pr}), and we extended it to metric perturbation resilience in Section \ref{sec:metric}.
In this section, we extend the result further to the local perturbation resilience setting, and we show how to add a worst-case guarantee of $O(\log^* n)$.
Specifically, we give a new algorithm, which is a natural modification to the $O(\log^* n)$ approximation algorithm of \cite{vishwanathan},
and show that it maintains the $O(\log^* n)$ guarantee in the worst case while returning each optimal perturbation resilient cluster in its own superset.
As a consequence, if the entire clustering instance satisfies 2-metric perturbation resilience, then the output of our algorithm is the optimal clustering.

\begin{theorem} \label{thm:asy_local}
Given an asymmetric $k$-center clustering instance $(S,d)$ of size $n$ with optimal clustering $\{C_1,\dots,C_k\}$,
for each 2-MPR cluster $C_i$, there exists a cluster outputted by Algorithm~\ref{alg:asy-pr} 
that is a superset of $C_i$ and does not contain any other 2-MPR cluster.
\footnote{Formally, given the output clustering $\mathcal{C}'$ of Algorithm \ref{alg:asy-pr}, for all 2-MPR clusters $C_i$ and $C_j$,
there exists $C_i'\in\mathcal{C}'$ such that $C_i\subseteq C_i'$ and $C_j\cap C_i'=\emptyset$.
}
Furthermore, the overall clustering returned by Algorithm \ref{alg:asy-pr} is an $O(\log^* n)$-approximation.
\end{theorem}

At the end of this section, we will also show an algorithm that outputs an optimal cluster $C_i$ exactly,
if $C_i$ and any optimal cluster near $C_i$ are 2-MPR.

\paragraph{Approximation algorithm for asymmetric $k$-center}
We start with a recap of the $O(\log^* n)$-approximation algorithm by \cite{vishwanathan}.
This was the first nontrivial algorithm for asymmetric $k$-center, 
and the approximation ratio was later proven to be tight by \citep{chuzhoy2005asymmetric}.
To explain the algorithm, it is convenient to think of asymmetric $k$-center as a set covering problem.
Given an asymmetric $k$-center instance $(S,d)$,
define the directed graph $D_{(S,d)}=(S,A)$, where $A=\{(u,v)\mid d(u,v)\leq r^*\}$.
For a point $v\in S$, we define $\Gamma_{\text{in}}(v)$ and $\Gamma_{\text{out}}(v)$ as the set of vertices with an arc to and from $v$, respectively,
in $D_{(S,d)}$.
The asymmetric $k$-center problem is equivalent to finding a subset $C\subseteq S$ of size
$k$ such that $\cup_{c\in C}\Gamma_{\text{out}}(c)=S$.
We also define $\Gamma_{\text{in}}^x(v)$ and $\Gamma_{\text{out}}^x(v)$ as the set of vertices which have a path of length $\leq x$ 
to and from $v$ in $D_{(S,d)}$, respectively, 
and we define $\Gamma_{\text{out}}^x(A)=\bigcup_{v\in A}\Gamma_{\text{out}}^x(v)$ for a set $A\subseteq S$, 
and similarly for $\Gamma_{\text{in}}^x(A)$.
It is standard to assume the value of $r^*$ is known; since it is one of $O(n^2)$ distances, 
the algorithm can search for the correct value in polynomial time.
\cite{vishwanathan} uses the following concept.

\begin{definition}
Given an asymmetric $k$-center clustering instance $(S,d)$,
a point $v\in S$ is a \emph{center-capturing vertex} (CCV) if 
$\Gamma_{\text{in}}(v)\subseteq \Gamma_{\text{out}}(v)$.
In other words, for all $u\in S$, $d(u,v)\leq r^*$ implies $d(v,u)\leq r^*$.
\end{definition}

\begin{figure}
    \centering
    \begin{subfigure}[b]{0.35\textwidth}
        \includegraphics[width=\textwidth]{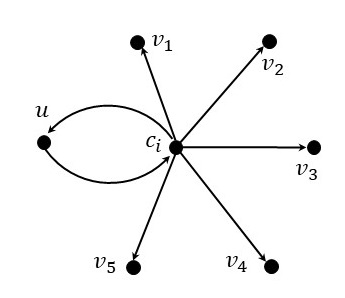}
        \caption{\small $u$ is a CCV, so it is distance $2r^*$ to its entire cluster.
				}
        \label{fig:ccv}
    \end{subfigure}	
		\quad\quad\quad\quad
    \begin{subfigure}[b]{0.35\textwidth}
        \includegraphics[width=\textwidth]{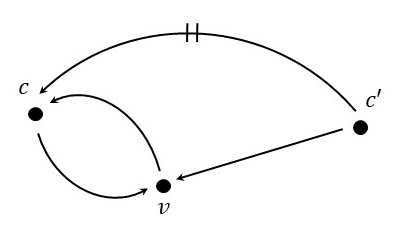}
				\caption{\small $c$ satisfies CCV-proximity, so it is closer to $v$ than $c'$ is to $v$. }
 \label{fig:proximity}
     
    \end{subfigure}
		    \caption{Examples of a center-capturing vertex (left), and CCV-proximity (right).}
		    \label{fig:approx-ccv}	
\end{figure}

As the name suggests, each CCV $v\in C_i$, ``captures'' its center, i.e.\ $c_i\in\Gamma_{\text{out}}(v)$ (see Figure~\ref{fig:ccv}).
Therefore, $v$'s entire cluster is contained inside $\Gamma_{\text{out}}^2(v)$, which is a nice property that the
approximation algorithm exploits.
At a high level, the approximation algorithm has two phases. In the first phase, the algorithm iteratively picks
a CCV $v$ arbitrarily and removes all points in $\Gamma_{\text{out}}^2(v)$. 
This continues until there are no more CCVs.
For every CCV picked, the algorithm is guaranteed to remove an entire optimal cluster.
In the second phase, the algorithm runs $\log^* n$ rounds of
a greedy set-cover subroutine on the remaining points. 
See Algorithm \ref{alg:asy}.
To prove the second phase terminates in $O(\log^* n)$ rounds, the analysis crucially assumes there are no
CCVs among the remaining points. We refer the reader to~\citep{vishwanathan} for these details.

\begin{algorithm}[h]
\caption {\textsc{$O(\log^* n)$ approximation algorithm for asymmetric $k$-center~\citep{vishwanathan}}}
\label{alg:asy}
\begin{algorithmic}
\Require{Asymmetric $k$-center instance $(S, d)$, optimal radius $r^*$ (or try all possible candidates)}
\Statex Set $C=\emptyset$
\Statex\textbf{Phase I: Pull out arbitrary CCVs}
\State While there exists an unmarked CCV
\begin{itemize}
\item Pick an unmarked CCV $c$, add $c$ to $C$, and mark all vertices in $\Gamma_{\text{out}}^{2}(c)$
\end{itemize}
\textbf{Phase II: Recursive set cover}
\State Set $A_0=S\setminus\Gamma_{\text{out}}^{5}(C)$, $i=0$. 
\State While $|A_i|>k$:
\begin{itemize}
\item Set $A'_{i+1}=\emptyset$. 
\item While there exists an unmarked point in $A_i$:
\begin{itemize}
\item Pick $v\in S$ which maximizes $\Gamma_{\text{out}}^{5}(v)\cap A_i$, mark points in $\Gamma_{\text{out}}^{5}(v)\cap A_i$, and add $v$ to $A'_{i+1}$. 
\end{itemize}
\item Set $A_{i+1}=A'_{i+1}\cap A_0$ and $i=i+1$
\end{itemize}
\Ensure{Centers $C\cup A_{i+1}$}
\end{algorithmic}
\end{algorithm}

\paragraph{Description of our algorithm and analysis}
We show a modification to the approximation algorithm of \cite{vishwanathan} leads to simultaneous guarantees
in the worst case and under local perturbation resilience.
Note that the set of all CCV's is identical to the symmetric set $A$ defined in Section \ref{sec:2pr-asymmetric}.
In Section \ref{sec:2pr-asymmetric}, we showed that all centers are in $A$, therefore, all centers are CCV's, assuming 2-PR.
In this section, we have that each 2-MPR center is a CCV (Lemma~\ref{lem:akc-properties}),
which is true by definition of $r^*$, ($C_i\subseteq\Gamma_{\text{out}}(c_i)$) and by using the definition of 2-MPR 
($\Gamma_{\text{in}}(c_i)\subseteq C_i$).
Since each 2-MPR center is a CCV, we might hope that we can output the 2-MPR clusters by
iteratively choosing a CCV $v$ and removing all points in $\Gamma_{\text{out}}^2(v)$. 
However, using this approach we might remove two or more 2-MPR centers in the same iteration, which means we would not output one separate
cluster for each 2-MPR cluster.
If we try to get around this problem by iteratively choosing a CCV $v$ and removing all points in $\Gamma_{\text{out}}^1(v)$,
then we may not remove one full cluster in each iteration, so for example, some of the 2-MPR clusters may be cut in half.

The key challenge is thus carefully specifying which nearby
points get marked by each CCV $c$ chosen by the algorithm.
We fix this problem with two modifications that carefully balance the two guarantees.
First, any CCV $c$ chosen will mark points in the following way: for all $c'\in \Gamma_{\text{in}}(c)$, mark all points in 
$\Gamma_{\text{out}}(c')$. Intuitively, we still mark points that are two hops from $c$, but the first `hop' must go backwards,
i.e., mark $v$ such that there exists $c'$ and $d(c',c)\leq r^*$ and $d(c',v)\leq r^*$.
This gives us a useful property:
if the algorithm picks a CCV $c\in C_i$ and it marks a different 2-MPR center $c_j$, then the middle hop
must be a point $q$ in $C_j$.
However, we know from perturbation resilience that $d(c_j,q)<d(c,q)$.
This fact motivates the final modification to the algorithm.
Instead of picking arbitrary CCVs,
we require the algorithm to choose CCVs with an extra structural property which we call \emph{CCV-proximity} (Definition~\ref{def:prox}).
See Figure \ref{fig:proximity}.
Intuitively, a point $c$ satisfying CCV-proximity must be closer than other CCVs to each point in $\Gamma_{\text{in}}(c)$.
Going back to our previous example, $c$ will NOT satisfy CCV-proximity because $c_j$ is closer to $q$, but we will be able to show that
all 2-MPR centers do satisfy CCV-proximity.
Thus Algorithm~\ref{alg:asy-pr} works as follows. It first chooses points satisfying CCV-proximity and marks points according to the rule
mentioned earlier. When there are no more points satisfying CCV-proximity, the algorithm chooses regular CCVs. 
Finally, it runs Phase II as in Algorithm~\ref{alg:asy}.
This ensures that Algorithm~\ref{alg:asy-pr} will output each 2-MPR center in its own cluster.

\paragraph{Details for Theorem~\ref{thm:asy_local}}
Now we formally define CCV-proximity.
The other properties in the following definition, \emph{center-separation}, and \emph{weak CCV-proximity},
are defined in terms of the optimal clustering, so they cannot be explicitly used by an algorithm,
but they will simplify all of our proofs.

\begin{definition}
\begin{enumerate}
\item An optimal center $c_i$ satisfies \emph{center-separation} if 
any point within distance $r^*$ of $c_i$ belongs to its cluster $C_i$. 
That is, $\Gamma_{\text{in}}(c_i)\subseteq C_i$ (see Figure~\ref{fig:center-sep}).
\footnote{
Center-separation is the local-PR equivalent of property 2 from Section \ref{sec:2pr-asymmetric}.
}
\item
A CCV $c\in C_i$ satisfies \emph{weak CCV-proximity} if,
given a CCV $c'\notin C_i$ and a point $v\in C_i$, we have $d(c,v)<d(c',v)$ (see Figure~\ref{fig:cluster-prox}).
\footnote{
This is a variant of $\alpha$-center proximity~\citep{awasthi2012center},
a property defined over an entire clustering instance, 
which states for all $i$, for all $v\in C_i$, $j\neq i$, we have $\alpha d(c_i,v)<d(c_j,v)$.
Our variant generalizes to local-PR, asymmetric instances, and general CCV's.
}
\item A point $c$ satisfies \emph{CCV-proximity} if it is a CCV, and each point in $\Gamma_{\text{in}}(c)$ 
is closer to $c$ than any CCV outside of $\Gamma_{\text{out}}(c)$. 
That is, for all points $v\in \Gamma_{\text{in}}(c)$ and CCVs $c'\notin \Gamma_{\text{out}}(c)$, 
$d(c,v)<d(c',v)$ (see Figure~\ref{fig:proximity})
\footnote{This is similar to a property in \citep{balcan2012clustering}.}
\end{enumerate}
 \label{def:prox}
\end{definition}

\begin{figure}
    \centering
    \begin{subfigure}[b]{0.40\textwidth}
        \includegraphics[width=\textwidth]{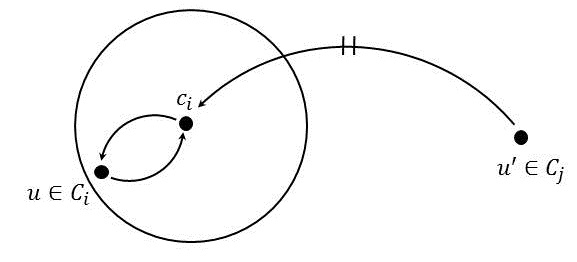}
        \caption{\small $c_i$ satisfies center-separation, so $\Gamma_{\text{in}}(c_i)\subseteq C_i$.
				}
        \label{fig:center-sep}
    \end{subfigure}	
		\quad\quad\quad\quad
    \begin{subfigure}[b]{0.40\textwidth}
        \includegraphics[width=\textwidth]{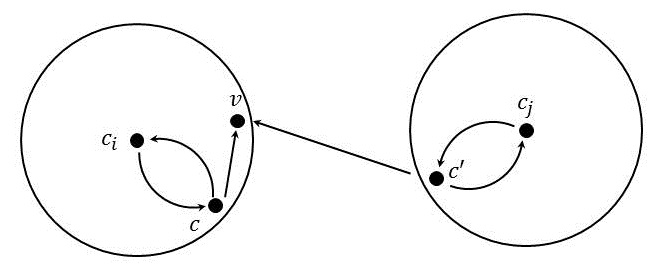}
				\caption{\small $c\in C_i$ satisfies weak CCV-proximity, so it is closer to $v\in C_i$ than $c'\in C_j$ is to $v$. }
 \label{fig:cluster-prox}
     
    \end{subfigure}
		    \caption{Examples of center-separation (left),
				and weak CCV-proximity (right).}
		    \label{fig:sep-prox}	
\end{figure}

Next we prove that all 2-MPR centers satisfy center-separation,
and all CCV's from a 2-MPR cluster satisfy CCV-proximity and weak CCV-proximity.

\begin{lemma}\label{lem:akc-properties}
Given an asymmetric $k$-center clustering instance $(S,d)$ and a 2-MPR cluster $C_i$, \\
\emph{(1)} $c_i$ satisfies center-separation, \\
\emph{(2)} any CCV $c\in C_i$ satisfies CCV-proximity, \\
\emph{(3)} any CCV $c\in C_i$ satisfies weak CCV-proximity.
\end{lemma}

\begin{proof} 

Given an instance $(S,d)$ and a 2-MPR cluster $C_i$,
we show that $C_i$ has the desired properties.

\emph{Center separation:}
Assume there exists a point $v\in C_j$ for $j\neq i$ such that $d(v,c_i)\leq r^*$.
The idea is to construct a $2$-perturbation in which $v$ becomes the center for $C_i$.

\begin{equation*}
d''(s,t)=
\begin{cases}
\min( 2 r^*,2 d(s,t) ) & \text{if } s=v \text{, } t\in C_i  \\
2 d(s,t) & \text{otherwise.}
\end{cases}
\end{equation*}
$d''$ is a valid 2-perturbation of $d$ because for each point $u\in C_i$, 
$d(v,u)\leq d(v,c_i)+d(c_i,u)\leq 2r^*$.
Define $d'$ as the metric completion of $d''$.
Then by Lemma \ref{lem:d'_metric}, $d'$ is a 2-metric perturbation with
optimal cost $2 r^*$.
The set of centers $\{c_{i'}\}_{i'=1}^k\setminus\{c_i\}\cup\{v\}$ achieves the optimal
cost, since $v$ is distance $2 r^*$ from $C_i$, and all other 
clusters have the same center as in $\mathcal{OPT}$ (achieving radius $2 r^*$).
If $v$ is a noncenter, then $\{c_{i'}\}_{i'=1}^k\setminus\{c_i\}\cup\{v\}$ is a valid set of $k$ centers.
If $v=c_j$, then add an arbitrary point $v'\in C_j$ to this set of centers (it still achieves the optimal cost since adding another center
can only decrease the cost).
Then in this new optimal clustering, $c_i$'s center is a point in $\{c_{i'}\}_{i'=1}^k\setminus\{c_i\}\cup\{v,v'\}$,
none of which are from $C_i$. We conclude that $C_i$ is no longer an optimal cluster, contradicting 2-MPR.

\emph{Weak CCV-proximity:}
Given a CCV $c\in C_i$, a CCV $c'\in C_j$ such that $j\neq i$, and a point $v\in C_i$,
assume to the contrary that $d(c',v)\leq d(c,v)$.
We will construct a perturbation in which $c$ and $c'$ become centers of their respective
clusters, and then $v$ switches clusters.
Define the following perturbation $d''$.

\begin{equation*}
d''(s,t)=
\begin{cases}
\min( 2 r^*, 2 d(s,t) ) & \text{if } s=c \text{, } t\in C_i\text{ or }s=c'\text{, }t\in
C_j\cup\{v\} \\
2 d(s,t) & \text{otherwise.}
\end{cases}
\end{equation*}
$d''$ is a valid 2-perturbation of $d$ because for each point $u\in C_i$,
$d(c,u)\leq d(c,c_i)+d(c_i,u)\leq 2r^*$, for each point $u\in C_j$,
$d(c',u)\leq d(c',c_j)+d(c_j,u)\leq 2r^*$, and $d(c',v)\leq d(c,v)\leq d(c,c_i)+d(c_i,v)\leq 2r^*$.
Define $d'$ as the metric completion of $d''$.
Then by Lemma \ref{lem:d'_metric}, $d'$ is a 2-metric perturbation with
optimal cost $2 r^*$.
The set of centers $\{c_{i'}\}_{i'=1}^k\setminus\{c_i,c_j\}\cup\{c,c'\}$ achieves the optimal
cost, since $c$ and $c'$ are distance $2r^*$ from $C_i$ and $C_j$, and all other 
clusters have the same center as in $\mathcal{OPT}$ (achieving radius $2 r^*$).
Then since $d'(c',v)\leq d(c,v)$, $v$ can switch clusters, contradicting perturbation resilience.

\emph{CCV-proximity:}
First we show that $c_i$ is a CCV.
By center-separation, we have that $\Gamma_{\text{in}}(c_i)\subseteq C_i$, and by definition of $r^*$,
we have that $C_i\subseteq\Gamma_{\text{out}}(c_i)$. Therefore, $\Gamma_{\text{in}}(c_i)\subseteq C_i\subseteq\Gamma_{\text{out}}(c_i)$,
so $c_i$ is a CCV.
Now given a point $v\in \Gamma_{\text{in}}(c_i)$ and a CCV $c\notin \Gamma_{\text{out}}(c_i)$,
from center-separation and definition of $r^*$, we have
$v\in C_i$ and $c\in C_j$ for $j\neq i$.
Then from weak CCV-proximity, 
$d(c_i,v)<d(c,v)$.
\end{proof}

\begin{algorithm}[h]
\caption {\textsc{Algorithm for asymmetric $k$-center under perturbation resilience}}
\label{alg:asy-pr}
\begin{algorithmic}
\Require{ Asymmetric $k$-center instance $(S, d)$, distance $r^*$ (or try all possible candidates)}
\State Set $C=\emptyset$. \\
\textbf{Phase I: Pull out special CCVs}
\begin{itemize}
\item While there exists an unmarked CCV:
\begin{itemize}
\item Pick an unmarked point $c$ which satisfies CCV-proximity.
If no such $c$ exists, then pick an arbitrary unmarked CCV instead. Add $c$ to $C$, and $\forall c'\in \Gamma_{\text{in}}(c)$, mark $\Gamma_{\text{out}}(c')$.
\end{itemize}
\item For each $c\in C$, let $V_c$ denote $c$'s Voronoi tile of the marked points
induced by $C$.
\end{itemize}
\textbf{Phase II: Recursive set cover}
\begin{itemize}
\item Run Phase II as in Algorithm~\ref{alg:asy}, outputting $A_{i+1}$.
\item Compute the Voronoi diagram $\{V_c'\}_{c\in C\cup A_{i+1}}$ of $S\setminus\Gamma^5_{\text{out}}(C)$ induced by $C\cup A_{i+1}$
\item For each $c$ in $C$, set $V'_c=V_c\cup V'_c$
\end{itemize}
\Ensure{Sets $\{V_c'\}_{c\in C\cup A_{i+1}}$}
\footnote{The final outputted clustering is not a Voronoi partition. See the proof of Theorem~\ref{thm:asy_local}.}
\end{algorithmic}
\end{algorithm}

Now using Lemma~\ref{lem:akc-properties}, we can prove Theorem~\ref{thm:asy_local}.

\begin{proof}[Proof of Theorem~\ref{thm:asy_local}]

First we explain why Algorithm \ref{alg:asy-pr} retains the approximation guarantee of Algorithm \ref{alg:asy}.
Given any CCV $c\in C_i$ chosen in Phase I, since $c$ is a CCV, then $c_i\in\Gamma_{\text{out}}(c)$, and by definition
of $r^*$, $C_i\subseteq\Gamma_{\text{out}}(c_i)$. Therefore, each chosen CCV always marks its cluster, and we start
Phase II with no remaining CCVs. 
This condition is sufficient for Phase II to return an
$O(\log^* n)$ approximation (Theorem 3.1 from~\cite{vishwanathan}).

Next we claim that for each 2-MPR cluster $C_i$, there exists a cluster outputted by Algorithm~\ref{alg:asy-pr} 
that is a superset of $C_i$ and does not contain any other 2-MPR cluster.
To prove this claim, we first show there exists a point from $C_i$ satisfying CCV-proximity that cannot be marked by any point from
a different cluster in Phase I.
From Lemma~\ref{lem:akc-properties}, $c_i$ satisfies CCV-proximity and
center-separation. If a point $c\notin C_i$ marks $c_i$, then
$\exists v\in\Gamma_{\text{in}}(c)\cap\Gamma_{\text{in}}(c_i)$.
By center-separation, $c_i\notin\Gamma_{\text{out}}(c)$, and therefore since $c$ is a CCV, $c\notin\Gamma_{\text{out}}(c_i)$.
But then from the definition of CCV-proximity for $c_i$ and $c$,
we have $d(c,v)<d(c_i,v)$ and $d(c_i,v)<d(c,v)$, so we have reached a contradiction (see Figure \ref{fig:marking}).

\begin{figure}
    \centering
    \begin{subfigure}[b]{0.35\textwidth}
        \includegraphics[width=\textwidth]{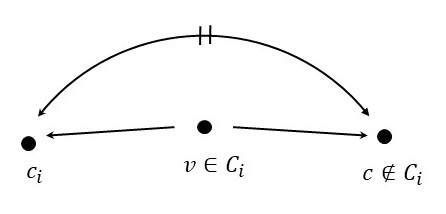}
        \caption{\small A point $c\notin C_i$ cannot mark $c_i$ without causing a contradiciton. }
        \label{fig:marking}
    \end{subfigure}	
		\quad\quad\quad\quad
    \begin{subfigure}[b]{0.45\textwidth}
        \includegraphics[width=\textwidth]{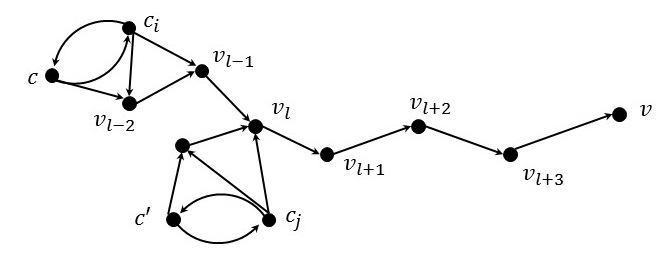}
				\caption{\small Case 2 in the proof of Theorem \ref{thm:strong}: $c'$ is closer to $v$ than $c$ is to $v$. }
 \label{fig:strong}
     
    \end{subfigure}
		    \caption{Example of Algorithm \ref{alg:asy-pr} (left), and the proof of Theorem \ref{thm:strong} (right).}
		    \label{fig:mark-strong}	
\end{figure}

At this point, we know a point $c\in C_i$ will always be chosen by the algorithm in Phase I.
To finish the proof, we show that each point $v$ from $C_i$ is closer to $c$ than to any other point $c'\notin C_i$ chosen as a center in Phase I.
Since $c$ and $c'$ are both CCVs, this follows directly from weak CCV-proximity.
\footnote{
It is possible that a center $c'$ chosen in Phase 2 may be closer to $v$ than $c$ is to $v$, causing $c'$ to ``steal'' $v$; this is unavoidable.
This is why Algorithm \ref{alg:asy-pr} separately computes the voronoi tiling from Phase I and Phase II,
and so the final output is technically not a valid voronoi tiling over the entire instance $S$.}
\end{proof}

\subsubsection{Strong local perturbation resilience}

Theorem \ref{thm:asy_local} shows that Algorithm~\ref{alg:asy-pr} will
output each 2-PR center in its own cluster.
Given some 2-PR center $c_i$, it is unavoidable that $c_i$ might mark
a \emph{non} 2-PR center $c_j$, and capture all points in its cluster.
In this section, we show that Algorithm \ref{alg:asy-pr} with a slight modification outputs each 
2-strong local perturbation resilient cluster exactly.
Recall that intuitively, an optimal cluster $C_i$ satisfies $\alpha$-strong local perturbation resilience if all nearby clusters satisfy
$\alpha$-perturbation resilience (definition~\ref{def:spr}).

\begin{figure}
    \centering
        \includegraphics[width=.5\textwidth]{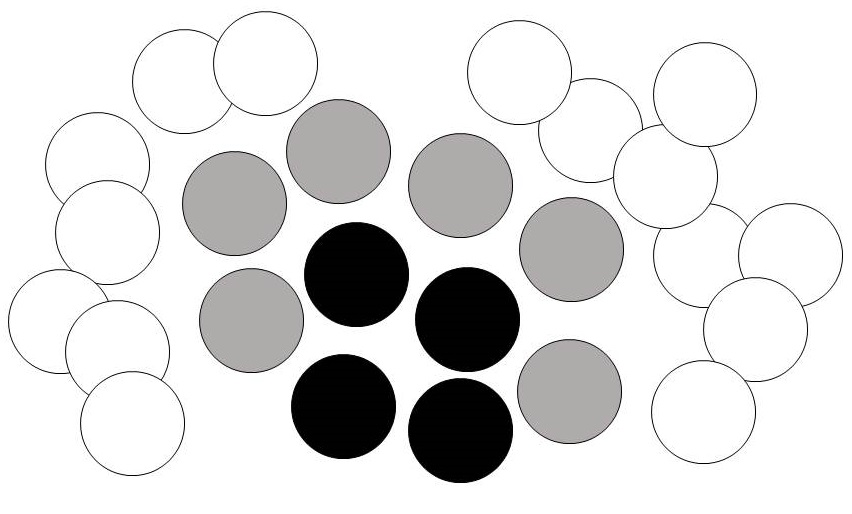}
        \caption{\small The white clusters are optimal clusters with no structure, the gray clusters are 2-PR clusters,
				and the black clusters are 2-PR clusters who only have neighbors which are also 2-PR (Theorem \ref{thm:strong}).
				Algorithm \ref{alg:asy-strong} outputs the black clusters exactly.}
        \label{fig:strong-example}
\end{figure}

Intuitively, the nearby 2-PR clusters `shield' $C_i$ from all other points (see Figure \ref{fig:strong-example}).
The only modification is that at the end of Phase II, instead of calculating the Voronoi diagram using the metric $d$,
we assign each point $v\in S\setminus\Gamma^5_{\text{out}}(C)$ to the point in $C\cup A_{i+1}$
which minimizes the path length in $D_{(S,d)}$,
breaking ties by distance to first common vertex in the shortest path.

\begin{algorithm}[h]
\caption {\textsc{Outputting optimal clusters for asymmetric $k$-center under stability}}
\label{alg:asy-strong}
\begin{algorithmic}
\Require{ Asymmetric $k$-center instance $(S, d)$, distance $r^*$ (or try all possible candidates)}
\State Set $C=\emptyset$. \\
\textbf{Phase I: Pull out special CCVs}
\begin{itemize}
\item While there exists an unmarked CCV:
\begin{itemize}
\item Pick an unmarked point $c$ which satisfies CCV-proximity.
If no such $c$ exists, then pick an arbitrary unmarked CCV instead. Add $c$ to $C$, and $\forall c'\in \Gamma_{\text{in}}(c)$, mark $\Gamma_{\text{out}}(c')$.
\end{itemize}
\item For each $c\in C$, let $V_c$ denote $c$'s Voronoi tile of the marked points
induced by $C$.
\end{itemize}
\textbf{Phase II: Recursive set cover}
\begin{itemize}
\item Run Phase II as in Algorithm~\ref{alg:asy}, outputting $A_{i+1}$.
\end{itemize}
\textbf{Phase III: Assign points to centers}
\begin{itemize}
\item For each $v\in S\setminus\Gamma^5_{\text{out}}(C)$, assign $v$ to the center $c\in C\cup A_{i+1}$
with the minimum path length in $D_{(S,d)}$ from $c$ to $v$, breaking ties by distance to first common vertex in the shortest path.
\item Let $V_c'$ denote the set of vertices in $v\in S\setminus\Gamma^5_{\text{out}}(C)$ assigned to $c$.
\item For each $c$ in $C$, set $V'_c=V_c\cup V'_c$
\end{itemize}
\Ensure{Sets $\{V_c'\}_{c\in C\cup A_{i+1}}$}
\end{algorithmic}
\end{algorithm}

\begin{theorem} \label{thm:strong}
Given an asymmetric $k$-center clustering instance $(S,d)$ with optimal clustering $\mathcal{C}=\{C_1,\dots,C_k\}$,
consider a 2-PR cluster $C_i$.
Assume that for all $C_j$ for which there is $v\in C_j$, $u\in C_i$, and $d(u,v)\leq r^*$, we have that $C_j$ is also 2-PR ($C_i$ satisfies 2-strong local perturbation resilience).
Then Algorithm~\ref{alg:asy-strong} returns $C_i$ exactly.
\end{theorem}

\begin{proof} 
Given a 2-PR cluster $C_i$ with the property in the theorem statement, by Theorem \ref{thm:asy_local},
there exists a CCV $c\in C_i$ from Phase I satisfying CCV-proximity such that $C_i\subseteq V_c$.
Our goal is to show that $V_c=C_i$.
First we show that $\Gamma_{\text{in}}(c)\subseteq C_i$, which will help us prove the theorem.
Assume towards contradiction that there exists a point $v\in\Gamma_{\text{in}}(c)\setminus C_i$.
Let $v\in C_j$. Since $c$ is a CCV, we have $v\in\Gamma_{\text{out}}(c)$, so $C_j$ must be 2-PR by definition.
By Lemma \ref{lem:akc-properties}, 
$c_j$ is a CCV and $d(c_j,v)<d(c,v)$. 
But this violates CCV-proximity of $c$, so we have reached a contradiction. 
Therefore, $\Gamma_{\text{in}}(c)\subseteq C_i$. 

To finish the proof, we must show that $V_c\subseteq C_i$.
Assume towards contradiction there exists $v\in V_c\setminus C_i$ at the end of the algorithm.

Case 1: $v$ was marked by $c$ in Phase I. Let $v\in C_j$.
Then there exists a point $u\in \Gamma_{\text{in}}(c)$ such that $v\in\Gamma_{\text{out}}(u)$.
From the previous paragraph we have that $\Gamma_{\text{in}}(c)\subseteq C_i$, so $u\in C_i$.
Therefore, $v\in\Gamma_{\text{out}}(u)$ implies $C_j$ must be 2-PR.
Since $v$ is from a different 2-PR cluster, it cannot be contained in $V_c$, so we have reached a contradiction.

Case 2: $v$ was not marked by $c$ in Phase I.
Denote the shortest path in $D_{(S,d)}$ from $c$ to $v$ by $(c=v_0,v_1,\dots,v_{L-1},v_L=v)$.
Let $v_\ell\in C_j$ denote the first vertex on the shortest path that is not in $C_i$ (such a vertex must exist because $v\notin C_i$).
Then $v_{\ell-1}\in C_i$ and $d(v_{\ell-1},v_\ell)\leq r^*$, so $C_j$ is 2-PR by the assumption in Theorem \ref{thm:strong}.
Let $c'$ denote the CCV chosen in Phase I such that $C_j\subseteq V_{c'}$.
Then by weak CCV-proximity from Lemma \ref{lem:akc-properties}, we have $d(c',v_\ell)<d(c,v_\ell)$.

Case 2a: $d(c,v_\ell)\leq r^*$. Then $v_\ell$ is the first vertex on the shortest path from $c$ to $v$ and $c'$ to $v$,
so $v_\ell$ is the first common vertex on the shortest paths. Since $d(c',v_\ell)<d(c,v_\ell)$, the algorithm will choose
$c'$ as the center for $v$.
See Figure~\ref{fig:strong}.

Case 2b: $d(c,v_\ell)> 2r^*$. But since $c'\in C_j$ is a CCV, we have $d(c',v_\ell)\leq d(c',c_j)+d(c_j,v_\ell)\leq 2r^*$, so the shortest path from
$c'$ to $v_\ell$ on $D_{(S,d)}$ is at most 2, and the shortest path from $c$ to $v_\ell$ on $D_{(S,d)}$ is at least 3. Since $v_\ell$ is on the shortest
path from $c$ to $v$, it follows that the shortest path from $c'$ to $v$ is strictly shorter than the shortest path from $c$ to $v$.

Case 2c: $r^*<d(c,v_\ell)\leq 2r^*$. In this case, we will show that $d(c',v_\ell)\leq r^*$,
and therefore we conclude the shortest path from $c'$ to $v$ is strictly shorter than the shortest path from $c$ to $v$, as in Case 2b.
Assume towards contradiction that $d(c',v_\ell)>r^*$.
Then we will create a 2-perturbation in which $c$ and $c'$ become centers for their own clusters, and $v_\ell$ switches clusters.
Define the following perturbation $d'$.

\begin{equation*}
d'(s,t)=
\begin{cases}
\min( 2 r^*, 2 d(s,t) ) & \text{if } s=c \text{, } t\in C_i\text{ or }s=c'\text{, }t\in C_j\setminus\{v_\ell\} \\
d(s,t) & \text{if } s=c \text{, }t=v_\ell \\
2 d(s,t) & \text{otherwise.}
\end{cases}
\end{equation*}
$d'$ is a valid 2-perturbation of $d$ because for each point $u\in C_i$, $d(c,u)\leq d(c,c_i)+d(c_i,u)\leq 2r^*$, 
for each point $u\in C_j$, $d(c',u)\leq d(c',c_j)+d(c_j,u)\leq 2r^*$, 
and $d(c,v_\ell)\leq 2r^*$. Therefore, $d'$ does not decrease any distances (and by construction, $d'$ does not increase
any distance by more than a factor of 2). 
If the optimal cost is $2r^*$, then the set of centers $\{c_{i'}\}_{i'=1}^k\setminus\{c_i,c_j\}\cup\{c,c'\}$ achieves the optimal
cost, since $c$ and $c'$ are distance $2r^*$ from $C_i\cup \{v_\ell\}$ and $C_j$, and all other 
clusters have the same center as in $\mathcal{OPT}$ (achieving radius $2 r^*$).
Then by perturbation resilience, it must be the case that $d'(c',v_\ell)<d'(c,v_\ell)$, which implies $2d(c',v_\ell)<d(c,v_\ell)$.
But $d(c',v_\ell)>r^*$ and $d(c,v_\ell)\leq 2r^*$, so we have a contradiction.
Now we assume the optimal cost of $d'$ is less than $2r^*$. 
Note that all distances $d(s,t)$ were increased to $2d(s,t)$ or $\min(2d(s,t),2r^*)$ except for $d(c,v_\ell)$.
Therefore, $c$ must be a center for $v_\ell$ under $d'$, or else the optimal cost could would be exactly $2r^*$ by Lemma \ref{lem:d'_metric}.
But it contradicts perturbation resilience to have $c$ and $c_\ell$ in the same optimal cluster under a 2-perturbation.
This completes the proof.
\end{proof}

\section{$k$-center under $(\alpha,\epsilon)$-perturbation resilience} \label{sec:3eps}

In this section, we consider $(\alpha,\epsilon)$-perturbation resilience.
First, we show that any 2-approximation algorithm for symmetric $k$-center must be optimal under  $(3,\epsilon)$-perturbation resilience (Theorem \ref{thm:3epsThm}).
Next, we show how to extend this result to local perturbation resilience (Theorem \ref{thm:3epsthm}).
Then we give an algorithm for asymmetric $k$-center which returns a clustering that is $\epsilon$-close to $\OPT$ under $(3,\epsilon)$-perturbation resilience (Theorem \ref{thm:3eps_asy}).
For all of these results, we assume a lower bound on the size of the optimal clusters, $|C_i|>2\epsilon n$ for all $i\in [k]$.
We show the lower bound on cluster sizes is necessary; in its absence,
the problem becomes $NP$-hard for all values of $\alpha\geq 1$ and $\epsilon>0$ (Theorem \ref{thm:apx}).
The theorems in this section require a careful reasoning about sets of centers under
different perturbations that cannot all simultaneously be optimal.

\subsection{Symmetric $k$-center}

We show that for any $(3,\epsilon)$-perturbation resilient $k$-center instance such that $|C_i|>2\epsilon n$ for all $i\in [k]$, 
there cannot be any pair of points from different clusters which are distance $\leq r^*$.
This structural result implies that simple algorithms will return the optimal clustering, such as running any 2-approximation algorithm
or running the Single Linkage algorithm, which is a fast algorithm widely used in practice for its simplicity.

\begin{theorem} \label{thm:3epsThm}
Given a $(3,\epsilon)$-perturbation resilient symmetric
$k$-center instance $(S,d)$ where all optimal clusters are size
$>\max(2\epsilon n,3)$, then the optimal clusters in $\OPT$ are exactly the
connected components of the threshold graph $G_{r^*}=(S,E)$,
where $E=\{(u,v)\mid d(u,v)\leq r^*\}$.
\end{theorem}

First we explain the high-level idea behind the proof.

\noindent{\bf Proof idea.}
Since each optimal cluster center is distance $r^*$ from all points in its cluster,
it suffices to show that any two points in different clusters are greater than $r^*$ apart
from each other.
Assume on the contrary that there exist $p\in C_i$ and $q\in C_{j\neq i}$
such that $d(p,q)\leq r^*$.
First we find a set of $k+2$ points and a 3-perturbation $d'$, such that every size $k$
subset of the points are optimal centers under $d'$. Then we show how this leads to a
contradiction under $(3,\epsilon)$-perturbation resilience.

Here is how we find a set of $k+2$ points and a perturbation $d'$ such that all size $k$
subsets are optimal centers under $d'$.
From our assumption, $p$ is distance $\leq 3r^*$ from every point in $C_i\cup C_j$
(by the triangle inequality).
Under a 3-perturbation in which all distances are blown up by a factor of 3
except $d(p,C_i\cup C_j)$, then replacing 
$c_i$ and $c_j$ with $p$ would still give us a set of $k-1$ centers that achieve the
optimal cost.
But, \emph{would this contradict $(3,\epsilon)$-perturbation resilience?} Not necessarily!
Perturbation resilience requires exactly $k$ \textit{distinct} centers.
\footnote{
This distinction is well-motivated; if for some application, the best $k$-center
solution is to put two centers at the same location, then we could achieve the
exact same solution with $k-1$ centers. That implies we should have been running
$k'$-center for $k'=k-1$ instead of $k$.}
The key challenge is to pick a final ``dummy'' center to guarantee
that the Voronoi partition is $\epsilon$-far from $\OPT$.
The dummy center might ``accidentally'' be the closest center for almost all points in $C_i$ or $C_j$.
Even worse, it might be the case that the new center sets off a chain reaction in which it becomes center to
a cluster $C_x$, and $c_x$ becomes center to $C_j$, which would also result in a partition that is not 
$\epsilon$-far from $\OPT$.

To deal with the chain reactions, 
we crucially introduce the notion of a
\emph{cluster capturing center} (CCC). 
A cluster capturing center (CCC) is not to be confused with a center-capturing vertex (CCV),
defined by \cite{vishwanathan} and used in the previous section.
$c_x$ is a CCC for $C_y$, if for all but $\epsilon n$ points $p\in C_y$,
$d(c_x,p)\leq r^*$ and for all $i\neq x,y$, $d(c_x,p)<d(c_i,p)$.
Intuitively, a CCC exists if and only if
$c_x$ is a valid center for $C_y$ when $c_y$ is taken out of the
set of optimal centers (i.e., a chain reaction will occur).
We argue that if a CCC does not exist then every dummy center we pick
must be close to either $C_i$ or $C_j$, since there are no chain reactions. 
If there does exist a CCC $c_x$ for $C_y$, then it is much harder to reason about what
happens to the dummy centers under $d'$, since there may be chain reactions.
However, we can define a new $d''$ by increasing
all distances except $d(c_x,C_y)$, which allows us to
take $c_y$ out of the set of optimal centers, 
and then any dummy center must be close to $C_x$ or $C_y$. There
are no chain reactions because we already know $c_x$ is the best center
for $C_y$ among the original optimal centers.
Thus, whether or not there exists a CCC, 
we can find $k+2$ points close to the entire dataset by 
picking points from both $C_i$ and $C_j$ (resp. $C_x$ and $C_y$).

Because of the assumption that all clusters are size $>2\epsilon n$, for every
3-perturbation there must be a bijection between clusters and centers, where
the center is closest to the majority of points in the corresponding cluster.
We show that all size $k$ subsets of the $k+2$ points cannot simultaneously
admit bijections that are consistent with one another.

\noindent{\bf Formal analysis.}
We start out with a simple implication from the assumption that
$|C_i|>2\epsilon n$ for all $i$.

\begin{fact} \label{fact:half}
Given a clustering instance which is $(\alpha,\epsilon)$-perturbation resilient
for $\alpha\geq 1$, and all optimal clusters have size $>2\epsilon n$,
then for any $\alpha$-perturbation $d'$, for any set of optimal centers
$c'_1,\dots,c'_k$ of $d'$, for each optimal cluster $C_i$, there must be a unique center
$c'_i$ which is the center for more than half of the points in $C_i$ under
$d'$.
\end{fact}

This fact follows simply from the definition of $(\alpha,\epsilon)$-perturbation resilience
(under $d'$, at most $\epsilon n$ points in the optimal solution can change clusters),
and the assumption that all optimal clusters are size $>2\epsilon n$.
Now we formally define a CCC.

\begin{definition} \label{def:CCC}
A center $c_i$ is a \textit{first-order cluster-capturing center} (CCC) 
for $C_j$ if for all $x\neq j$,
for more than half of the points $p\in C_j$, $d(c_i,p)< d(c_x,p)$ and
$d(c_i,p)\leq r^*$ (see Figure \ref{fig:ccc}).
$c_i$ is a \textit{second-order cluster-capturing center} (CCC2) 
for $C_j$ if there exists a $c_l$ such that for all $x\neq j,l$,
for more than half of points $p\in C_j$, $d(c_i,p)< d(c_x,p)$
and $d(c_i,p)\leq r^*$ (see Figure \ref{fig:CCC2}).
\end{definition}

\begin{figure}
\centering
  \hspace*{\fill}%
\begin{subfigure}[b]{0.22\textwidth}
\includegraphics[width =\textwidth]{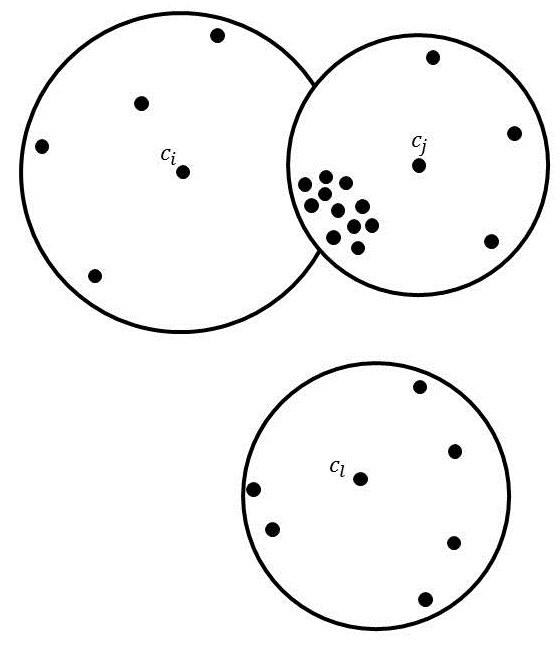}
\caption{\small \sl $c_i$ is a CCC for $C_j$.}
\label{fig:ccc}
\end{subfigure}
~
\hfill
\begin{subfigure}[b]{0.22\textwidth}
\includegraphics[width =\textwidth]{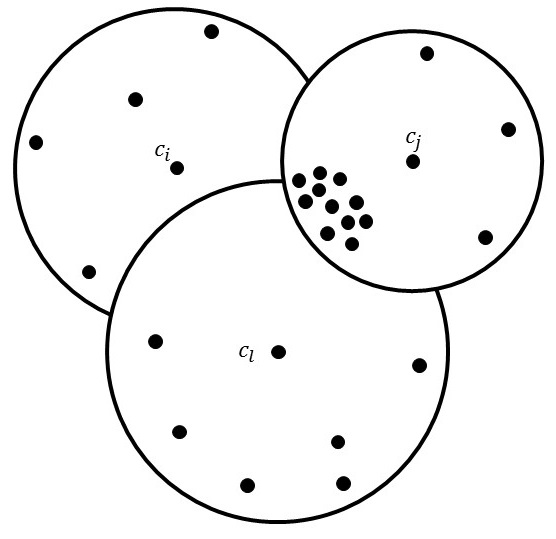}
\caption{\small \sl $c_i$ is a CCC for $C_j$ and $c_x$ is a CCC2 for $C_j$.}
\label{fig:CCC2}
\end{subfigure}
  \hspace*{\fill}%
\caption{(a)  Definition of a CCC, and (b) definition of a CCC2.}
\end{figure}

Each cluster $C_j$ can have at most one CCC $c_i$ because
$c_i$ is closer than any other center to more than half of $C_j$.
Every CCC is a CCC2, since the former is a stronger condition.
However, it is possible for a cluster to have multiple CCC2's.
\footnote{In fact, a cluster can have at most three CCC2's, but we do not use this in our analysis.}
We needed to define CCC2 for the following reason. Assuming there exist
$p\in C_i$ and $q\in C_j$ which are close, and we replace $c_i$ and
$c_j$ with $p$ in the set of centers. 
It is possible that $c_j$ is a CCC for $C_i$, but this does not help us, since we want to analyze the
set of centers after removing $c_j$.
However, if we know that $c_x$ is a CCC2 for $C_i$ (it is the best center for $C_i$, 
disregarding $c_j$),
then we know that $c_x$ will be the best center for $C_i$ after replacing
$c_i$ and $c_j$ with $p$.
Now we use this definition to show that if two points from different clusters are
close, we can find a set of $k+2$ points and a 3-perturbation $d'$, such that every size $k$
subset of the points are optimal centers under $d'$. 
To formalize this notion, we give one more definition.

\begin{definition} \label{def:hit}
A set $C\subseteq S$ $(\beta,\gamma)$-hits $S$ if for all $s\in S$,
there exist $\beta$ points in $C$ at distance $\leq\gamma r^*$ to $s$.
\end{definition}

Note that if a set $C$ of $k+2$ points $(3,3)$-hits $S$,
then any size $k$ subset of $C$ is still $3r^*$ from every point in $S$,
and later we will show that means there exists a perturbation $d'$
such that every size $k$ subset must be an optimal set of centers.

\begin{lemma} \label{lem:closepoints}
Given a clustering instance satisfying $(3,\epsilon)$-perturbation resilience such that
all optimal clusters are size $>2\epsilon n$
and there are two points from different clusters which are $\leq r^*$ apart from each other,
then there exists a set $C\subseteq S$ of size $k+2$ which
$(3,3)$-hits $S$.
\end{lemma}

\begin{proof}
First we prove the lemma assuming that a CCC2 exists, and then we prove the
other case. When a CCC2 exists, we do not need the assumption that two
points from different clusters are close.

Case 1: There exists a CCC2.
If there exists a CCC, then denote $c_x$ as a CCC for $C_y$.
If there does not exist a CCC, then denote $c_x$ as a CCC2 for $C_y$.
We will show that all points are close to either $C_x$ or $C_y$.
$c_x$ is distance $\leq r^*$ to all but
$\epsilon n$ points in $C_y$. Therefore, $d(c_x,c_y)\leq 2r^*$ and so $c_x$ is
distance $\leq 3r^*$ to all points in $C_y$. Consider the following $d'$.

\begin{equation*}
d'(s,t)=
\begin{cases}
\min( 3 r^*, 3 d(s,t) ) & \text{if } s=c_x \text{, } t\in C_y  \\
3 d(s,t) & \text{otherwise.}
\end{cases}
\end{equation*}

This is a $3$-perturbation because $d(c_x,t)\leq 3 r^*$ for all $t\in C_y$.
Then by Lemma~\ref{lem:d'_same_cost}, the optimal cost is $3 r^*$.
Given any $u\in S$,
the set of centers $\{c_l\}_{i=1}^k\setminus\{c_y\}\cup\{u\}$ achieves the optimal
cost, since $c_x$ is distance $3 r^*$ from $C_y$, and all other 
clusters have the same center as in $\mathcal{OPT}$ (achieving radius $3 r^*$). 
Therefore, this set of centers must create a partition that is
$\epsilon$-close to $\mathcal{OPT}$, or else there would be a contradiction.
Then from Fact~\ref{fact:half}, one of the centers in 
$\{c_l\}_{i=1}^k\setminus\{c_y\}\cup\{u\}$ must be the center for
the majority of points in $C_y$ under $d'$.
If this center is $c_\ell$, $\ell\neq x,y$, 
then for the majority of points $v\in C_y$, 
$d(c_\ell,v)\leq r^*$ and $d(c_\ell,v)< d(c_z,v)$ for all $z\neq \ell,y$.
Then by definition, $c_\ell$ is a CCC for $C_y$.
But then $\ell$ must equal $x$, so we have a contradiction.
Note that if some $c_\ell$ has for the majority of $v\in C_y$,
$d(c_\ell,v)\leq d(c_z,v)$ (non-strict inequality)
for all $z\neq \ell,y$, then there is another equally
good partition in which $c_\ell$ is not the center for the majority of points
in $C_y$, so we still obtain a contradiction.
Therefore, either $u$ or $c_x$ must be the center for the majority of points
in $C_y$ under $d'$.

If $c_x$ is the center for the majority of points in $C_y$, then $u$ must be the center
for the majority of points in $C_x$ (it cannot be a different 
center $c_\ell$, since $c_x$ is a better
center for $C_x$ than $c_\ell$ by definition).
Therefore, each $u\in S$ is distance $\leq r^*$ to all but $\epsilon n$ points in either $C_x$ or $C_y$.

Now partition all the non-centers into two sets $S_x$ and $S_y$, such that
$S_x=\{u'\mid \text{for the majority of points }v'\in C_x,\text{ }d(u',v')\leq r^*\}$ and
$S_y=\{u'\mid u'\notin S_x\text{ and for the majority of points }v'\in C_y,\text{ }d(u',v')\leq r^*\}$.

Then given $u',v'\in S_x$,
there exists an $s\in C_x$ such that
$d(u',v')\leq d(u',s)+d(s,v')\leq 2r^*$ 
(since both points are close to more than half of points in $C_x$).
Similarly, any two points $u',v'\in S_y$ are $\leq 2r^*$ apart.
See Figure \ref{fig:ccc-arg}.

Now we will find a set of $k+2$ points that $(3,3)$-hits $S$.
For now, assume that $S_x$ and $S_y$ are both nonempty.
Given an arbitrary pair $p\in S_x$, $q\in S_y$, 
we claim that $\{c_\ell\}_{\ell=1}^k\cup\{p,q\}$ $(3,3)$-hits $S$.
Given a non-center $s\in C_i$ such that $i\neq x$ and $i\neq y$, without loss of generality let $s\in S_x$.
Then $c_i$, $p$, and $c_x$ are all distance $3r^*$ to $s$.
Furthermore, $c_i$, $p$ and $c_x$ are all distance $3r^*$ to $c_i$.
Given a point $s\in C_x$, then $c_x$, $c_y$, and $p$ are distance $3r^*$ to $s$ because $d(c_x,c_y)\leq 2r^*$,
and a similar argument holds for $s\in C_y$.
Therefore, $\{c_\ell\}_{\ell=1}^k\cup\{p,q\}$ $(3,3)$-hits $S$.

If $S_x=\emptyset$ or $S_y=\emptyset$, then we can prove a slightly stronger statement:
for each pair of non-centers $\{p,q\}$, $\{c_\ell\}_{\ell=1}^k\cup\{p,q\}$ $(3,3)$-hits $S$.
Without loss of generality, let $S_y=\emptyset$.
Given a point $s\in C_i$ such that $i\neq x$ and $i\neq y$, then $c_i$, $c_x$, and $p$ are all distance $3r^*$ to $s$.
Given a point $s\in C_x$, then $p$, $c_x$, and $c_y$ are all distance $\leq 3r^*$ to $s$.
Given a point $s\in C_y$, then $p$, $c_x$, and $c_y$ are all distance $\leq 3r^*$ to $s$ because $s,p\in S_x$ implies $d(s,p)\leq 2r^*$.
Thus, we have proven case 1.

\begin{figure}
    \centering
\includegraphics[width =.35\textwidth]{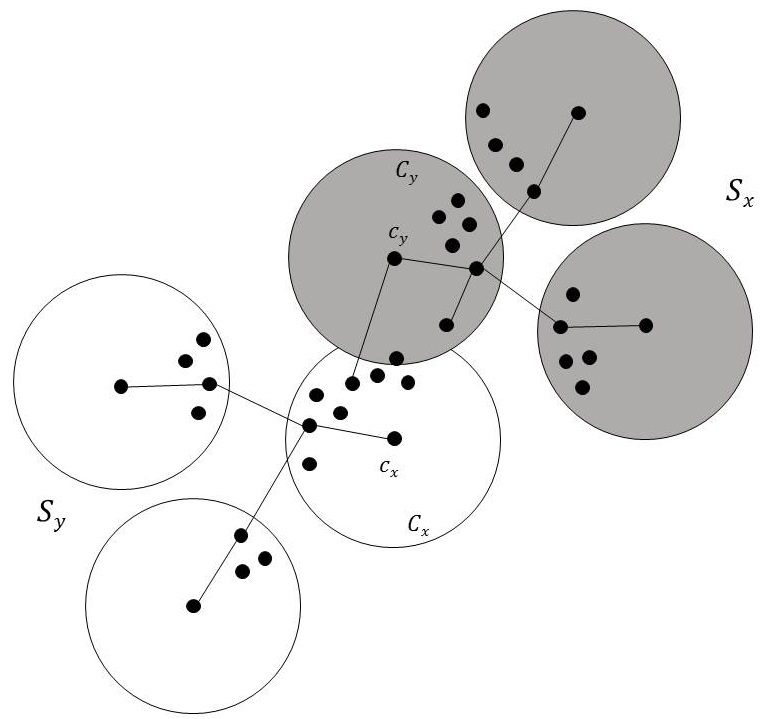}
\caption{\small \sl Case 1 of Lemma \ref{lem:closepoints}}
\label{fig:ccc-arg}
\end{figure}

Case 2: There does not exist a CCC2. Now we use the assumption that there
exist $p\in C_x$, $q\in C_y$, $x\neq y$, such that $d(p,q)\leq r^*$.
Then by the triangle inequality, $p$ is distance $\leq 3r^*$
to all points in $C_x$ and $C_y$.
Consider the following $d'$.

\begin{equation*}
d'(s,t)=
\begin{cases}
\min( 3 r^*, 3 d(s,t) ) & \text{if } s=p \text{, } t\in C_x\cup C_y  \\
3 d(s,t) & \text{otherwise.}
\end{cases}
\end{equation*}

This is a $3$-perturbation because $d(p,t)\leq 3 r^*$ for all $t\in C_x\cup C_y$.
Then by Lemma~\ref{lem:d'_same_cost}, the optimal cost is $3 r^*$.
Given any $s\in S$,
the set of centers $\{c_l\}_{i=1}^k\setminus\{c_x,c_y\}\cup\{p,s\}$ achieves the optimal
cost, since $p$ is distance $3 r^*$ from $C_x\cup C_y$, and all other 
clusters have the same center as in $\mathcal{OPT}$ (achieving radius $3 r^*$). 
Therefore, this set of centers must create a partition that is
$\epsilon$-close to $\mathcal{OPT}$, or else there would be a contradiction.
Then from Fact~\ref{fact:half}, one of the centers in 
$\{c_l\}_{i=1}^k\setminus\{c_x,c_y\}\cup\{p,s\}$ must be the center for
the majority of points in $C_x$ under $d'$.

If this center is $c_\ell$ for $\ell\neq x$ and $\ell\neq y$, 
then for the majority of points $t\in C_x$, 
$d(c_\ell,t)\leq r^*$ and $d(c_\ell,t)< d(c_z,t)$ for all $z\neq \ell,x,y$.
Then by definition, $c_\ell$ is a CCC2 for $C_x$, and we have a contradiction.
Note that if some $c_\ell$ has for the majority of $t\in C_x$, 
$d(c_\ell,t)\leq d(c_z,t)$ (non-strict inequality)
for all $z\neq \ell,y$, then there is another equally
good partition in which $c_\ell$ is not the center for the majority of points
in $C_y$, so we still obtain a contradiction.

Similar logic applies to the center for the majority of points in $C_y$.
Therefore, $p$ and $s$ must be the centers for $C_x$ and $C_y$.
Since $s$ was an arbitrary noncenter, all noncenters are 
distance $\leq r^*$ to all but $\epsilon n$ points in either $C_x$ or $C_y$.

Similar to Case 1, we now partition all the non-centers into two sets $S_x$ and $S_y$, such that
$S_x=\{u\mid \text{for the majority of points }v\in C_x,\text{ }d(u,v)\leq r^*\}$ and
$S_y=\{u\mid u\notin S_x\text{ and for the majority of points }v\in C_y,\text{ }d(u,v)\leq r^*\}$.
As before, each pair of points in $S_x$ are distance $\leq 2r^*$ apart, and similarly for $S_y$.
It is no longer true that $d(c_x,c_y)\leq 2r^*$, however, we can prove that for both $S_x$ and $S_y$,
there exist points from two distinct clusters each. From the previous paragraph, given a non-center
$s\in C_i$ for $i\neq x,y$, we know that $p$ and $s$ are centers for $C_x$ and $C_y$.
With an identical argument, given $t\in C_j$ for $j\neq x,y,i$, we can show that $q$ and $t$ are centers for $C_x$ and $C_y$.
It follows that $S_x$ and $S_y$ both contain points from at least two distinct clusters.

Now we finish the proof by showing that for each pair $u\in S_x$, $v\in S_y$, $\{c_\ell\}_{\ell=1}^k\cup\{u,v\}$ $(3,3)$-hits $S$.
Given a non-center $s\in C_i$, without loss of generality $s\in S_x$, then there exists $j\neq i$ and $t\in C_j\cap S_x$.
Then $c_i$, $c_j$, and $u$ are $3r^*$ to $s$ and $c_i$, $c_x$, and $u$ are $3r^*$ to $c_i$.
In the case where $i=x$, then $c_i$, $c_j$, and $u$ are $3r^*$ to $c_i$.
This concludes the proof.
\end{proof}

So far, we have shown that by just assuming two points from different
clusters are close, we can find a set of $k+2$ points that
$(3,3)$-hits $S$.
Now we will show that such a set leads to a contradiction under 
$(3,\epsilon)$-perturbation resilience.
Specifically, we will show there exists a perturbation $d'$ such that
any size $k$ subset can be an optimal set of centers.
But it is not possible that all ${k+2 \choose k}$ of these sets of centers simultaneously create partitions that are $\epsilon$-close to $\OPT$.
First we state a lemma which proves there does exist a perturbation $d'$
such that any size $k$ subset is an optimal set of centers.

\begin{lemma} \label{lem:hit}
Given a $k$-center clustering instance $(S,d)$, given $z\geq 0$,
and given a set $C\subseteq S$ of size $k+z$ which $(z+1,\alpha)$-hits $S$,
there exists an $\alpha$-metric perturbation $d'$ such that all size $k$ subsets of $C$ are optimal sets of centers under $d'$.
\end{lemma}

\begin{proof}
Consider the following perturbation $d''$.

\begin{equation*}
d''(s,t)=
\begin{cases}
\min( \alpha r^*, \alpha d(s,t) ) & \text{if } s\in C \text{ and } d(s,t)\leq \alpha r^*  \\
\alpha d(s,t) & \text{otherwise.}
\end{cases}
\end{equation*}

By Lemma \ref{lem:d'_metric}, the metric closure $d'$ of $d''$ 
is an $\alpha$-metric perturbation with optimal cost $\alpha r^*$. 
Given any size $k$ subset $C'\subseteq C$, 
then for all $v\in S$, there is still at least one $c\in C'$ such that
$d(c,v)\leq \alpha r^*$, therefore by construction, $d'(c,v)\leq \alpha r^*$.
It follows that $C'$ is a set of optimal centers under $d'$.
\end{proof}

Next, we state a fact that helps clusters rank their best centers from
the set of $k+2$ points.
For each cluster $C_i$, we would like to have a ranking of all points
such that for a given $d'$ and set of $k$ centers, the center for
$C_i$ is the highest point in the ranking.
The following fact shows this ranking is well-defined.

\begin{fact} \label{fact:ranking}
Given a $k$-center clustering instance $(S,d)$
with optimal clustering $\mathcal{C}=\{C_1,\dots,C_k\}$
such that for all $i\in [k]$, $|C_i|>2\epsilon n$,
let $d'$ denote an $\alpha$-perturbation of $d$.
There are rankings $R_{x,d'}$ for all $x\in [k]$
such that for any optimal set of centers $c'_1,\dots,c'_k$ under $d'$, 
the center that is closest in $d'$ to all
but $\epsilon n$ points in $C_x$ is the highest-ranked point in $R_{x,d'}$.
\footnote{Formally, for each $C_x$, there exists a bijection $R_{x,d'}:S\rightarrow[n]$
such that for all sets of $k$ centers $C$ that achieve the optimal cost under $d'$, 
we have $c=\text{argmin}_{c'\in C} R_{x,d'}(c')$ if and only if $\text{Vor}_C(c)$ is $\epsilon$-close to $C_x$.
}
\end{fact}

\begin{proof}
Assume the fact is false. Then there exists a $d'$, a cluster $C_i$, two points $p$ and $q$, and two
sets of $k$ centers $p,q\in C$ and $p,q\in C'$ which achieve the optimal
cost under $d'$, but $p$ is the center for
$C_i$ in $C$ while $q$ is the center for $C_i$ in $C'$.
Then $p$ is closer than all other points in $C$ to all but $\epsilon n$ points in $C_i$.
Similarly, $q$ is closer than all other points in $C'$ to all but $\epsilon n$ points in $C_i$.
Since $|C_i|>2\epsilon n$, this causes a contradiction.
\end{proof}

We also define $R_{x,d',C}:C\rightarrow [n']$ as the ranking specific to a set of centers $C$, where $|C|=n'$.
Now we can prove Theorem \ref{thm:3epsThm}.

\begin{proof} [Proof of Theorem \ref{thm:3epsThm}]
It suffices to prove that any two points from different clusters are
at distance $>r^*$ from each other.
Assume towards contradiction that this is not true.
Then by Lemma~\ref{lem:closepoints},
there exists a set $C$ of size $k+2$ which $(3,3)$-hits $S$.
From Lemma \ref{lem:hit}, there exists a 3-metric perturbation $d'$
such that all size $k$ subsets of $C$ are optimal sets of centers under $d'$.
Consider the ranking of each cluster for $C$ over $d'$ guaranteed
from Fact~\ref{fact:ranking}.
We will show this ranking leads to a contradiction.

Consider the set of all points ranked 1 or 2
by any cluster, formally,
$\{p\in C\mid \exists i\text{ s.t. }R_{i,d',C}\leq 2\}$.
This set is a subset of $C$, since we are only considering the rankings
of points in $C$, so it is size $\leq k+2$.
Note that a point cannot be ranked both 1 and 2 by a cluster.
Then as long as $k>2$, it follows by the Pigeonhole Principle that there exists
a point $c\in C$ which is ranked in the top two by two different clusters.
Formally, there exists $x$ and $y$ such that $x\neq y$, 
$R_{x,d',C}(c)\leq 2$, and $R_{y,d',C}(c)\leq 2$.
Denote $u$ and $v$ such that $R_{x,d',C}(u)=1$ and $R_{y,d',C}(v)=1$.
If $u$ or $v$ is equal to $c$, then redefine it to an arbitrary center
in $C\setminus\{c,u,v\}$.
Consider the set of centers $C'=C\setminus\{u,v\}$ which is optimal under $d'$ by construction.
But then from Fact \ref{fact:ranking}, $c$ is the center for all but $\epsilon n$ 
points in both $C_x$ and $C_y$, contradicting Fact \ref{fact:half}.
This completes the proof.
\end{proof}

\subsection{Local perturbation resilience}

Now we extend the argument from the previous section to local perturbation resilience.
First we state our main structural result, 
which is that any pair of points from different $(3,\epsilon)$-PR clusters must be distance $>r^*$ from each other.
Then we will show how the structural result easily leads to an algorithm for $(3,\epsilon)$-SLPR clusters.

\begin{theorem} \label{thm:3epsthm}
Given a $k$-center clustering instance $(S,d)$ with optimal radius $r^*$ such that all optimal clusters are size $>2\epsilon n$
and there are at least three $(3,\epsilon)$-PR clusters,
then for each pair of $(3,\epsilon)$-PR clusters $C_i$ and $C_j$,
for all $u\in C_i$ and $v\in C_j$, we have $d(u,v)>r^*$. 
\end{theorem}

Before we prove this theorem, we show how it implies an algorithm to output the optimal $(3,\epsilon)$-SLPR clusters exactly.
Since the distance from each point to its closest center is $\leq r^*$, a corollary 
of Theorem~\ref{thm:3epsthm} is that
any 2-approximate solution must contain the optimal $(3,\epsilon)$-SLPR clusters,
as long as the 2-approximation satisfies two sensible conditions:
\emph{(1)} for every point $v$ and its assigned center $u$ (so we know $d(u,v)\leq 2r^*$),
$\exists w$ such that $d(u,w)$ and $d(w,v)$ are $\leq r^*$, and \emph{(2)}
there cannot be multiple clusters outputted in the 2-approximation that can be combined into one cluster with radius smaller than $r^*$.
Both of these properties are easily satisfied using quick pre- or post-processing steps.
\footnote{
For condition \emph{(1)}, before running the algorithm, remove all edges of distance $>r^*$, and then take the metric completion of
the resulting graph.
For condition \emph{(2)}, given the radius $\hat r$ of the
outputted solution, for each $v\in S$, check if the ball of radius $\hat r$ around $v$ captures multiple clusters. If so, combine them.}

\begin{theorem} \label{thm:3eps-full}
Given a $k$-center clustering instance $(S,d)$ such that all optimal clusters are size $>2\epsilon n$
and there are at least three $(3,\epsilon)$-PR clusters,
then any 2-approximate solution satisfying conditions \emph{(1)} and \emph{(2)} must contain all optimal $(3,\epsilon)$-SLPR clusters.
\end{theorem}

\begin{proof}
Given such a clustering instance, then Theorem~\ref{thm:3epsthm} ensures that there is no edge of length $r^*$ between 
points from two different $(3,\epsilon)$-PR clusters. Given a $(3,\epsilon)$-SLPR cluster $C_i$, it follows that there is no
point $v\notin C_i$ such that $d(v,C_i)\leq r^*$.
Therefore, given a 2-approximate solution $\mathcal{C}$ satisfying condition \emph{(1)}, 
any $u\in C_i$ and $v\notin C_i$ cannot be in the same cluster. 
This is because in the graph of datapoints where edges signify a distance $\leq r^*$, $C_i$ is an isolated component.
Finally, by condition \emph{(2)}, $C_i$ must not be split into two clusters.
Therefore, $C_i\in\mathcal{C}$.
\end{proof}

\paragraph{Proof idea for Theorem~\ref{thm:3epsthm}}
The high level idea of this proof is similar to the proof of Theorem \ref{thm:3epsThm}.
In fact, the first half is very similar to Lemma \ref{lem:closepoints}:
we show that if two points from
different PR clusters are close together, then there must exist a set of $k+2$ points $C$
which $(3,3)$-hits the entire point set.
In the previous section, we arrived at a contradiction by showing that it is not possible
that all ${k+2 \choose k}$ subsets of $C$ can be centers that are $\epsilon$-close to $\OPT$.
However, the weaker local PR assumption poses a new challenge.

As in the previous section, we will still argue that all size $k$ subsets of $C$ cannot stay consistent with the $(3,\epsilon)$-PR clusters
using a ranking argument which maps optimal clusters to optimal centers, 
but our argument will be to establish conditional claims which narrow down the possible sets of ranking lists.
For instance, assume there is a $(3,\epsilon)$-PR cluster $C_i$ which ranks $c_i$ first, and ranks $c_j$ second.
Then under subsets $C'$ which do not contain $c_i$, $c_j$ is the center for a cluster $C_i'$ which is $\epsilon$-close to $C_i$.
Therefore, a different point in $C'$ must be the center for all but $\epsilon n$ points in $C_j$ (and it cannot be a different center $c_\ell$
without causing a contradiction).
This is the basis for Lemma~\ref{lem:rank-struct}, which is the main workhorse lemma in the proof of Theorem~\ref{thm:3epsthm}.
By building up conditional statements, we are able to analyze every possibility of the ranking lists for the three
$(3,\epsilon)$-PR clusters and show that all of them lead to contradictions, proving Theorem~\ref{thm:3epsthm}.

\paragraph{Formal analysis of Theorem \ref{thm:3epsthm}}

We start with a local perturbation resilience variant of Fact \ref{fact:half}.

\begin{fact} \label{fact:half-local}
Given a $k$-center clustering instance $(S,d)$
such that all optimal clusters have size $>2\epsilon n$,
let $d'$ denote an $\alpha$-perturbation with optimal
centers $C'=\{c'_1,\dots,c'_k\}$.
Let $\mathcal{C'}$ denote the set of 
$(\alpha,\epsilon)$-PR clusters.
Then there exists a one-to-one function $f:\mathcal{C'}\rightarrow C'$
such that for all $C_i\in\mathcal{C'}$, 
$|\text{Vor}_{C,d'}(f(C_i))\cap C_i|\geq |C_i|-\epsilon n$.
That is, the optimal cluster in $d'$ whose center is $f(C_i)$ contains all but $\epsilon n$
of the points in $C_i$.
\end{fact}

In words, for any set of optimal centers under an $\alpha$-perturbation, 
each PR cluster can be paired to a unique center.
This follows simply because all optimal clusters are size $>2\epsilon n$,
yet under a perturbation, $<\epsilon n$ points can switch out of each PR cluster.
Because of this fact, for a perturbation $d'$ with set of optimal centers $C$ and an $(\alpha,\epsilon)$-PR cluster $C_x$,
we will say that $c$ is the center for $C_x$ under $d'$ if $c$ is the center for all but $\epsilon n$ points in $C_x$.
Now we are ready to prove the first half of Theorem \ref{thm:3epsthm},
stated in the following lemma.
The proof is similar to Lemma \ref{lem:closepoints}.

\begin{lemma} \label{lem:closepoints-local}
Given a $k$-center clustering instance $(S,d)$
such that all optimal clusters are size $>2\epsilon n$
and there exist two points at distance $r^*$
from different $(3,\epsilon)$-PR clusters,
then there exists a partition $S_x\cup S_y$ of the non-centers $S\setminus\{c_\ell\}_{\ell=1}^k$ 
such that for all pairs $p\in S_x$, $q\in S_y$, $\{c_\ell\}_{\ell=1}^k\cup\{p,q\}$ $(3,3)$-hits $S$.
\end{lemma}

\begin{proof}
This proof is split into two main cases. The first case is the following: there exists a CCC2 for a
$(3,\epsilon)$-PR cluster, disregarding a $(3,\epsilon)$-PR cluster.
In fact, in this case, we do not need the assumption that two points from different PR clusters are close.
If there exists a CCC to a $(3,\epsilon)$-PR cluster,
denote the CCC by $c_x$ and the cluster by $C_y$. Otherwise, let $c_x$ denote a CCC2
to a $(3,\epsilon)$-PR cluster $C_y$, disregarding a $(3,\epsilon)$-PR center $c_z$.
Then $c_x$ is at distance $\leq r^*$ to all but
$\epsilon n$ points in $C_y$. Therefore, $d(c_x,c_y)\leq 2r^*$ and so $c_x$ is at
distance $\leq 3r^*$ to all points in $C_y$. Consider the following perturbation $d''$.

\begin{equation*}
d''(s,t)=
\begin{cases}
\min( 3 r^*, 3 d(s,t) ) & \text{if } s=c_x \text{, } t\in C_y  \\
3 d(s,t) & \text{otherwise.}
\end{cases}
\end{equation*}

This is a $3$-perturbation because for all $v\in C_y$, $d(c_x,v)\leq 3 r^*$.
Define $d'$ as the metric completion of $d''$.
Then by Lemma \ref{lem:d'_metric}, $d'$ is a 3-metric perturbation with
optimal cost $3 r^*$.
Given any non-center $v\in S$,
the set of centers $\{c_\ell\}_{\ell=1}^k\setminus\{c_y\}\cup\{v\}$ achieves the optimal
score, since $c_x$ is at distance $3 r^*$ from $C_y$, and all other 
clusters have the same center as in $\mathcal{OPT}$ (achieving radius $3 r^*$). 
Therefore, from Fact~\ref{fact:half-local}, one of the centers in 
$\{c_\ell\}_{\ell=1}^k\setminus\{c_y\}\cup\{v\}$ must be the center for
all but $\epsilon n$ points in $C_y$ under $d'$.
If this center is $c_\ell$, $\ell\neq x,y$, 
then for all but $\epsilon n$ points $u\in C_y$, 
$d(c_\ell,u)\leq r^*$ and $d(c_\ell,u)< d(c_z,u)$ for all $z\neq \ell,y$.
Then by definition, $c_\ell$ is a CCC for the $(3,\epsilon)$-PR cluster, $C_y$.
But then by construction, $\ell$ must equal $x$, so we have a contradiction.
Note that if some $c_\ell$ has for all but $\epsilon n$ points $u\in C_y$,
$d(c_\ell,u)\leq d(c_z,u)$ (non-strict inequality)
for all $z\neq \ell,y$, then there is another equally
good partition in which $c_\ell$ is not the center for all but $\epsilon n$ points
in $C_y$, so we still obtain a contradiction.
Therefore, either $v$ or $c_x$ must be the center for all but $\epsilon n$ points
in $C_y$ under $d'$.

If $c_x$ is the center for all but $\epsilon n$ points in $C_y$, 
then because $C_y$ is $(3,\epsilon)$-PR, the corresponding cluster must contain
fewer than $\epsilon n$ points from $C_x$. 
Furthermore, since for all $\ell\neq x$ and $u\in C_x$, $d(u,c_x)<d(u,c_\ell)$,
it follows that $v$ must be the center
for all but $\epsilon n$ points in $C_x$.
Therefore, every non-center $v\in S$ is at distance $\leq r^*$ to all but $\epsilon n$ points in
either $C_x$ or $C_y$.

Now partition all the non-centers into two sets $S_x$ and $S_y$, such that
$S_x=\{p\mid \text{for the majority of points }q\in C_x,\text{ }d(p,q)\leq r^*\}$ and
$S_y=\{p\mid p\notin S_x\text{ and for the majority of points }q\in C_y,\text{ }d(p,q)\leq r^*\}$.

Then given $p,q\in S_x$,
there exists an $s\in C_x$ such that
$d(p,q)\leq d(p,s)+d(s,q)\leq 2r^*$ 
(since both points are close to more than half of points in $C_x$).
Similarly, any two points $p,q\in S_y$ are $\leq 2r^*$ apart.

Now we will find a set of $k+2$ points that $(3,3)$-hits $S$.
For now, assume that $S_x$ and $S_y$ are both nonempty.
Given a pair $p\in S_x$, $q\in S_y$, we claim that $\{c_\ell\}_{\ell=1}^k\cup\{p,q\}$ $(3,3)$-hits $S$.
Given a non-center $s\in C_i$ such that $i\neq x,y$, without loss of generality let $s\in S_x$.
Then $c_i$, $p$, and $c_x$ are all distance $3r^*$ to $s$.
Furthermore, $c_i$, $c_x$ and $p$ are all distance $3r^*$ to $c_i$.
Given a point $s\in C_x$, then $c_x$, $c_y$, and $p$ are distance $3r^*$ to $s$ because $d(c_x,c_y)\leq 2r^*$.
Finally, $c_x$, $c_y$ and $p$ are distance $3r^*$ to $c_x$, and similar arguments hold for $s\in C_y$ and $c_y$.
Therefore, $\{c_\ell\}_{\ell=1}^k\cup\{p,q\}$ $(3,3)$-hits $S$.

If $S_x=\emptyset$ or $S_y=\emptyset$, then we can prove a slightly stronger statement:
for each pair of non-centers $\{p,q\}$, $\{c_\ell\}_{\ell=1}^k\cup\{p,q\}$ $(3,3)$-hits $S$.
Without loss of generality, let $S_y=\emptyset$.
Given a point $s\in C_i$ such that $i\neq x$ and $i\neq y$, then $c_i$, $c_x$, and $p$ are all distance $3r^*$ to $s$.
Given a point $s\in C_x$, then $p$, $c_x$, and $c_y$ are all distance $\leq 3r^*$ to $s$.
Given a point $s\in C_y$, then $p$, $c_x$, and $c_y$ are all distance $\leq 3r^*$ to $s$ because $s,p\in S_x$ implies $d(s,p)\leq 2r^*$.
Thus, we have proven case 1.

Now we turn to the other case.
Assume there does not exist a CCC2 to a PR cluster, disregarding a PR center. 
In this case, we need to use the assumption that there exist $(3,\epsilon)$-PR clusters $C_x$ and $C_y$,
and $p\in C_x$, $q\in C_y$ such that $d(p,q)\leq r^*$.
Then by the triangle inequality, $p$ is distance $\leq 3r^*$
to all points in $C_x$ and $C_y$.
Consider the following $d''$.

\begin{equation*}
d''(s,t)=
\begin{cases}
\min( 3 r^*, 3 d(s,t) ) & \text{if } s=p \text{, } t\in C_x\cup C_y  \\
3 d(s,t) & \text{otherwise.}
\end{cases}
\end{equation*}

This is a $3$-perturbation because $d(p,v)\leq 3 r^*$ for all $v\in C_x\cup C_y$.
Define $d'$ as the metric completion of $d''$.
Then by Lemma \ref{lem:d'_metric}, $d'$ is a 3-metric perturbation with
optimal cost $3 r^*$.
Given any non-center $s\in S$,
the set of centers $\{c_\ell\}_{\ell=1}^k\setminus\{c_x,c_y\}\cup\{p,s\}$ achieves the optimal
cost, since $p$ is distance $3 r^*$ from $C_x\cup C_y$, and all other 
clusters have the same center as in $\mathcal{OPT}$ (achieving radius $3 r^*$). 

From Fact~\ref{fact:half-local}, one of the centers in 
$\{c_\ell\}_{\ell=1}^k\setminus\{c_x,c_y\}\cup\{p,s\}$ must be the center for
all but $\epsilon n$ points in $C_x$ under $d'$.
If this center is $c_\ell$ for $\ell\neq x,y$, 
then for all but $\epsilon n$ points $t\in C_x$, 
$d(c_\ell,t)\leq r^*$ and $d(c_\ell,t)< d(c_z,t)$ for all $z\neq \ell,x,y$.
So by definition, $c_\ell$ is a CCC2 for $C_x$ disregarding $c_y$, which contradicts our assumption.
Similar logic applies to the center for all but $\epsilon n$ points in $C_y$.
Therefore, $p$ and $s$ must be the centers for $C_x$ and $C_y$.
Since $s$ was an arbitrary non-center, all non-centers are 
distance $\leq r^*$ to all but $\epsilon n$ points in either $C_x$ or $C_y$.

Similar to Case 1, we now partition all the non-centers into two sets $S_x$ and $S_y$, such that
$S_x=\{u\mid \text{for the majority of points }v\in C_x,\text{ }d(u,v)\leq r^*\}$ and
$S_y=\{u\mid u\notin S_x\text{ and for the majority of points }v\in C_y,\text{ }d(u,v)\leq r^*\}$.
As before, each pair of points in $S_x$ are distance $\leq 2r^*$ apart, and similarly for $S_y$.
It is no longer true that $d(c_x,c_y)\leq 2r^*$, however, we can prove that for both $S_x$ and $S_y$,
there exist points from two distinct clusters each. From the previous paragraph, given a non-center
$s\in C_i$ for $i\neq x,y$, we know that $p$ and $s$ are centers for $C_x$ and $C_y$.
With an identical argument, given $t\in C_j$ for $j\neq x,y,i$, we can show that $q$ and $t$ are centers for $C_x$ and $C_y$.
It follows that $S_x$ and $S_y$ both contain points from at least two distinct clusters.

Now we finish the proof by showing that for each pair $u\in S_x$, $v\in S_y$, $\{c_\ell\}_{\ell=1}^k\cup\{u,v\}$ $(3,3)$-hits $S$.
Given a non-center $s\in C_i$, without loss of generality $s\in S_x$, then there exists $j\neq i$ and $t\in C_j\cap S_x$.
Then $c_i$, $c_j$, and $u$ are $3r^*$ to $s$ and $c_i$, $c_x$, and $u$ are $3r^*$ to $c_i$.
In the case where $i=x$, then $c_i$, $c_j$, and $u$ are $3r^*$ to $c_i$.
This concludes the proof.
\end{proof}

Now we move to the second half of the proof of Theorem \ref{thm:3epsthm}.
Recall that the proof from the previous section relied on a ranking argument,
in which optimal clusters were mapped to their closest centers from the set $C$ of $k+2$ points from the first half of the proof.
This is the basis for the following fact.

\begin{fact} \label{fact:ranking-local}
Given a $k$-center clustering instance $(S,d)$
with optimal clustering $\mathcal{C}=\{C_1,\dots,C_k\}$
such that for all $i\in [k]$, $|C_i|>2\epsilon n$,
let $d'$ denote an $\alpha$-perturbation of $d$
and let $\mathcal{C'}$ denote the set of 
$(\alpha,\epsilon)$-PR clusters.
For each $C_x\in\mathcal{C'}$, there exists a ranking $R_{x,d'}$ of $S$ such that for any set of optimal centers $C=\{c'_1,\dots,c'_k\}$ under $d'$,
the center that is closest in $d'$ to all but $\epsilon n$ points in $C_x$
is the highest-ranked point in $R_{x,d'}$.
\footnote{Formally, for each $C_x\in\mathcal{C'}$, there exists a bijection $R_{x,d'}:S\rightarrow[n]$
such that for all sets of $k$ centers $C$ that achieve the optimal cost under $d'$, 
then $c=\text{argmin}_{c'\in C} R_{x,d'}(c')$ if and only if $\text{Vor}_C(c)$ is $\epsilon$-close to $C_x$.
}
\end{fact}

\begin{proof}
Assume the lemma is false. Then there exists an $(\alpha,\epsilon)$-PR cluster $C_i$, 
two distinct points $u,v\in S$, and two
sets of $k$ centers $C$ and $C'$ both containing $u$ and $v$, and both sets achieve the optimal score under an $\alpha$-perturbation
$d'$, but $u$ is the center for $C_i$ in $C$ while $v$ is the center for $C_i$ in $C'$.
Then $\text{Vor}_C(u)$ is $\epsilon$-close to $C_i$; similarly, $\text{Vor}_{C'}(v)$ is $\epsilon$-close to $C_i$.
This implies $u$ is closer to all but $\epsilon n$ points in $C_i$ than $v$, and $v$
 is closer to all but $\epsilon n$ points in $C_i$ than $u$.
Since $|C_i|>2\epsilon n$, this causes a contradiction.
\end{proof}

We also define $R_{x,d',C}:C\rightarrow [n']$ as the ranking specific to $C$.
Recall that our goal is to show a contradiction assuming two points from different PR clusters are close.
From Lemma~\ref{lem:hit} and Lemma~\ref{lem:closepoints-local},
we know there is a set of $k+2$ points, and any size $k$ subset is optimal under a suitable perturbation.
By Lemma~\ref{fact:half-local}, each size $k$ subset must have a mapping from PR clusters to centers, and from Fact~\ref{fact:ranking-local}, these
mappings are derived from a ranking of all possible center points by the PR clusters.
In other words, each PR cluster $C_x$ can rank all the points in $S$, so that for any set of optimal
centers for an $\alpha$-perturbation, the top-ranked center is the one whose cluster is $\epsilon$-close to $C_x$.
Now, using Fact~\ref{fact:ranking-local}, we can try to give a contradiction by showing that there is no
set of rankings for the PR clusters that is consistent with all the optimal sets of centers guaranteed by Lemmas~\ref{lem:hit} and~\ref{lem:closepoints-local}.
The following lemma gives relationships among the possible rankings.
These will be our main tools for contradicting PR and thus finishing the proof of Theorem~\ref{thm:3epsthm}.

\begin{lemma} \label{lem:rank-struct}
Given a $k$-center clustering instance $(S,d)$
such that all optimal clusters are size $>2\epsilon n$, 
and given non-centers $p,q\in S$ such that
$C=\{c_\ell\}_{\ell=1}^k\cup\{p,q\}$ $(3,3)$-hits $S$,
let the set $\mathcal{C'}$ denote the set of $(3,\epsilon)$-PR clusters.
Consider a 3-perturbation $d'$ such that all size $k$ subsets of $C$ are optimal sets of centers under $d'$.
The following are true.
\begin{enumerate}
\item Given $C_x\in \mathcal{C'}$ and $C_i$ such that $i\neq x$, $R_{x,d'}(c_x)<R_{x,d'}(c_i)$.
\item There do not exist $s\in C$ and $C_x,C_y\in\mathcal{C'}$ such that $x\neq y$, and $R_{x,d',C}(s)+R_{y,d',C}(s)\leq 4$.
\item Given $C_i$ and $C_x\in\mathcal{C'}$ such that $x\neq i$,
if $R_{x,d',C}(c_i)\leq 3$, then for all $C_y\in\mathcal{C'}$ such that $y\neq x,i$,
$R_{y,d',C}(p)\geq 3$ and $R_{y,d',C}(q)\geq 3$.
\end{enumerate}
\end{lemma}

\begin{proof}
\begin{enumerate}
\item By definition of the optimal clusters, for each $s\in C_x$, $d(c_x,s)<d(c_i,s)$, and therefore by construction, $d'(c_x,s)<d'(c_i,s)$.
It follows that $R_{x,d'}(c_x)<R_{x,d'}(c_i)$.
\item Assume there exists $s\in C$ and $C_x,C_y\in \mathcal{C'}$ such that $R_{x,d',C}(s)+R_{y,d',C}(s)\leq 4$.

Case 1: $R_{x,d',C}(s)=1$ and $R_{y,d',C}(s)\leq 3$.
Define $u$ and $v$ such that $R_{y,d',C}(u)=1$ and $R_{y,d',C}(v)=2$. 
(If $u$ or $v$ is equal to $s$, then redefine it to an arbitrary center in $C\setminus\{s,u,v\}$.)
Consider the set of centers $C'=C\setminus\{u,v\}$ which is optimal under $d'$ by Lemma \ref{lem:hit}. 
By Fact~\ref{fact:ranking-local}, $s$ is the center for all but $\epsilon n$ points in both $C_x$ and $C_y$, causing a contradiction.

Case 2: $R_{x,d',C}(s)= 2$ and $R_{y,d',C}(s)= 2$.
Define $u$ and $v$ such that $R_{x,d',C}(u)=1$ and $R_{y,d',C}(v)=1$.
(Again, if $u$ or $v$ is equal to $s$, then redefine it to an arbitrary center in $C\setminus\{s,u,v\}$.) 
Consider the set of centers $C'=C\setminus\{u,v\}$ which is optimal under $d'$ by Lemma \ref{lem:hit}. 
However, by Fact~\ref{fact:ranking-local}, $s$ is the center for all but $\epsilon n$ points in both $C_x$ and $C_y$, causing a contradiction.
\item Assume $R_{x,d',C}(c_i)\leq 3$.
 
Case 1: $R_{x,d',C}(c_i)=2$. Then by Lemma~\ref{lem:rank-struct} part 1, 
$R_{x,d',C}(c_x)=1$.
Consider the set of centers $C'=C\setminus\{c_x,p\}$, which is optimal under $d'$.
By Fact \ref{fact:ranking-local}, $\text{Vor}_{C'}(c_i)$ must be $\epsilon$-close to $C_x$.
In particular, $\text{Vor}_{C'}(c_i)$ cannot contain more than $\epsilon n$ points from $C_i$. 
But by definition, for all $j\neq i$ and $s\in C_i$, $d(c_i,s)<d(c_j,s)$. It follows that
$\text{Vor}_{C'}(q)$ must contain all but $\epsilon n$ points from $C_i$.
Therefore, for all but $\epsilon n$ points $s\in C_i$, for all $j$, $d'(q,s)<d'(c_j,s)$.
If $R_{y,d',C}(q)\leq 2$, then $C_y$ ranks $c_y$ or $p$ number one.
Then for the set of centers $C'=C\setminus\{c_y,p\}$,
$\text{Vor}_{C'}(q)$ contains more than $\epsilon n$
points from $C_y$ and $C_i$, contradicting the fact that $C_y$ is $(3,\epsilon)$-PR.
Therefore, $R_{y,d',C}(q)\geq 3$. The argument to show $R_{y,d',C}(p)\geq 3$ is symmetric.

Case 2: $R_{x,d',C}(c_i)=3$. If there exists $j\neq i,x$ such that $R_{x,d',C}(c_i)=2$, then without loss of generality we
are back in case 1. By Lemma~\ref{lem:rank-struct} part 1, $R_{x,d',C}(c_x)\leq 2$. 
Then either $p$ or $q$ are ranked top two, without loss of generality $R_{x,d',C}(p)\leq 2$.
Consider the set $C'=C\setminus\{c_x,p\}$.
Then as in the previous case, $\text{Vor}_{C'}(c_i)$ must be $\epsilon$-close to 
$C_x$, implying for all but $\epsilon n$ points $s\in C_i$, for all $j$, $d'(q,s)<d'(c_j,s)$.
If $R_{y,d',C}(q)\leq 2$, again, $C_y$ ranks $c_y$ or $p$ as number one.
Let $C'=C\setminus\{c_y,p\}$, and then $\text{Vor}_{C'}(q)$ contains more than $\epsilon n$
points from $C_y$ and $C_i$, causing a contradiction.
Furthermore, if $R_{y,d',C}(p)\leq 2$, then we arrive at a contradiction by
Lemma~\ref{lem:rank-struct} part 2.
\end{enumerate}
\end{proof}

We are almost ready to bring everything together to give a contradiction. Recall that Lemma~\ref{lem:closepoints-local}
allows us to choose a pair $(p,q)$ such that $\{c_\ell\}_{\ell=1}^k\cup\{p,q\}$ $(3,3)$-hits $S$.
For an arbitrary choice of $p$ and $q$, we may not end up with a contradiction.
It turns out, we will need to make sure one of the points comes from a PR cluster, and is very high in the ranking
list of its own cluster. This motivates the following fact, which is the final piece to the puzzle.

\begin{fact} \label{fact:prox}
Given a $k$-center clustering instance $(S,d)$
such that all optimal clusters are size $>2\epsilon n$, 
given an $(\alpha,\epsilon)$-PR cluster $C_x$, and given $i\neq x$,
then there are fewer than $\epsilon n$ points $s\in C_x$ such that $d(c_i,s)\leq\min(r^*,\alpha d(c_x,s))$.
\end{fact}

\begin{proof}
Assume the fact is false. Then let $B\subseteq C_x$ denote a set of size $\epsilon n$ such that for all $s\in B$, 
$d(c_i,s)\leq\min(r^*,\alpha d(c_x,s))$.
Construct the following perturbation $d'$. For all $s\in B$, set $d'(c_x,s)=\alpha d(c_x,s)$. 
For all other pairs $s,t$, set $d'(s,t)=d(s,t)$. This is clearly an $\alpha$-perturbation by construction.
Then the original set of optimal centers still achieves cost $r^*$ under $d'$ because for all $s\in B$, $d'(c_i,s)\leq r^*$.
Clearly, the optimal cost under $d'$ cannot be $<r^*$. It follows that the original set of optimal centers $C$ is still optimal under $d'$.
However, all points in $B$ are no longer in $\text{Vor}_C(c_x)$ under $d'$, contradicting the fact that $C_x$ is $(\alpha,\epsilon)$-PR.
\end{proof}

Now we are ready to prove Theorem~\ref{thm:3epsthm}.

\begin{proof} [Proof of Theorem~\ref{thm:3epsthm}]
Assume towards contradiction that there are two points at distance $\leq r^*$ from different
$(3,\epsilon)$-PR clusters.
Then by Lemma~\ref{lem:closepoints-local},
there exists a partition $S_1,S_2$ of non-centers of $S$ such that for all 
pairs $p\in S_1$, $q\in S_2$, $\{c_\ell\}_{\ell=1}^k\cup\{p,q\}$ $(3,3)$-hit $S$.
Given three $(3,\epsilon)$-PR clusters $C_x$, $C_y$, and $C_z$,
let $c_x'$, $c_y'$, and $c_z'$ denote the centers in $\{c_1,\dots,c_k\}$ 
ranked highest by $C_x$, $C_y$, and $C_z$ disregarding
$c_x$, $c_y$, and $c_z$, respectively.
Define $p=\text{argmin}_{s\in C_x} d(c_x,s)$, and without loss of generality let $p\in S_1$. Then pick an arbitrary point $q$ from $S_2$, and define $C=\{c_\ell\}_{\ell=1}^k\cup\{p,q\}$.
Define $d'$ as in Lemma \ref{lem:hit}
(i.e., we define $d'$ so that all size $k$ subsets of $C$ are optimal sets of centers under $d'$).
We claim that $R_{x,d',C}(p)<R_{x,d',C}(c_x')$:
from Fact~\ref{fact:prox}, there are fewer than $\epsilon n$ points $s\in C_x$ such that $d(c_x',s)\leq\min(r^*,3 d(c_x,s))$.
Among each remaining point $s\in C_x$, we will show $d'(p,s)\leq d'(c_x',s)$. Recall that $d(p,s)\leq d(p,c_x)+d(c_x,s)\leq 2r^*$, so $d'(p,s)=\min(3r^*,3d(p,s))$. 
There are two cases to consider.

Case 1: $d(c_x',s)>r^*$. Then by construction, $d'(c_x',s)\geq 3r^*$, and so $d'(p,s)\leq d'(c_x',s)$.

Case 2: $3d(c_x,s)<d(c_x',s)$. 
Then 
\begin{align*}
d'(p,s)&\leq 3d(p,s)\text{ by construction of $d'$}\\
&\leq 3(d(p,c_x)+d(c_x,s))\text{ by triangle inequality}\\
&\leq 6d(c_x,s)\text{ by definition of $p$}\\
&\leq 2d(c_x',s)\text{ by assumption}\\
&\leq \min(3r^*,3d(c_x',s))\text{ by construction of $d'$}\\
&=d'(c_x',s),
\end{align*}
and this proves our claim.

Because $R_{x,d',C}(p)<R_{x,d',C}(c_x')$, and $R_{x,d',C}(c_x)<R_{x,d',C}(c_x')$,
it follows that the top two can only be $c_x$, $p$, or $q$.
Therefore, either $R_{x,d',C}(p)\leq 2$ or $R_{x,d',C}(q)\leq 2$.
The rest of the argument is broken up into cases.

Case 1: $R_{x,d',C}(c_x')\leq 3$. 
From Lemma~\ref{lem:rank-struct} part 3, then $R_{y,d',C}(p)\geq 3$ and $R_{y,d',C}(q)\geq 3$.
It follows by process of elimination that 
$R_{y,d',C}(c_{y})=1$ and $R_{y,d',C}(c_{y'})=2$.
Again by Lemma~\ref{lem:rank-struct} part 3, $R_{x,d',C}(p)\geq 3$ and $R_{x,d',C}(q)\geq 3$, causing a contradiction. 

Case 2: $R_{x,d',C}(c_{x'})>3$ and $R_{y,d',C}(c_{y'})\leq 3$.
Then $R_{x,d',C}(p)\leq 3$ and $R_{x,d',C}(q)\leq 3$. 
From Lemma~\ref{lem:rank-struct} part 3, $R_{x,d',C}(p)\geq 3$ and $R_{x,d',C}(q)\geq 3$, therefore we have a contradiction.
Note, the case where $R_{x,d',C}(c_{x'})>3$ and $R_{z,d',C}(c_{z'})\leq 3$ is identical to this case.

Case 3: The final case is when $R_{x,d',C}(c_{x'})>3$, $R_{y,d',C}(c_{y'})>3$, and $R_{z,d',C}(c_{z'})>3$.
So for each $i\in\{x,y,z\}$, the top three for $C_i$ in $C$ is a permutation of $\{c_i,p,q\}$.
Then each $i\in\{x,y,z\}$ must rank $p$ or $q$ in the top two, so by the Pigeonhole Principle, 
either $p$ or $q$ is ranked top two by
two different PR clusters, contradicting Lemma~\ref{lem:rank-struct}.
This completes the proof.
\end{proof}

We note that Case 3 in Theorem~\ref{thm:3epsthm} is the reason why we need to assume there are at least three
$(3,\epsilon)$-PR clusters. If there are only two, $C_x$ and $C_y$, it is possible that there exist
$u\in C_x$, $v\in C_y$ such that $d(u,v)\leq r^*$. 
In this case, for $p,q,d',$ and $C$ as defined in the proof of
Theorem~\ref{thm:3epsthm}, if $C_x$ ranks $c_x$, $p$, $q$ as its top three and $C_y$ ranks $c_y$, $q$, $p$
as its top three, then there is no contradiction.

\subsection{Asymmetric $k$-center}
Now we consider asymmetric $k$-center under $(3,\epsilon)$-PR.
The asymmetric case is a more challenging setting, and our algorithm does not return the optimal
solution, however, our algorithm outputs a clustering that is $\epsilon$-close to the optimal solution.

Recall the definition of the symmetric set $A$ from Section \ref{sec:2pr},
$A=\{p\mid \forall q, d(q,p)\leq r^* \implies d(p,q)\leq r^*\}$, 
equivalently, the set of all CCV's.
We might first ask whether $A$ respects the structure of $\OPT$,
as it did under 2-perturbation resilience. 
Namely, whether \emph{Condition 1:} all optimal centers are in $A$, and \emph{Condition 2:}
$\arg\min_{q\in A}d(q,p)\in C_i\implies p\in C_i$ hold.
In fact, we will show that neither conditions hold in the asymmetric case,
but both conditions are only slightly violated.

\subsubsection{Structure of optimal centers}
First we give upper and lower bounds on the number of optimal centers in $A$,
which will help us construct an algorithm for $(3,\epsilon)$-PR later on.
We call a center $c_i$ ``bad'' if it is not in the set $A$, i.e., 
$\exists q$ such that $d(q,c_i)\leq r^*$ but $d(c_i,q)>r^*$.
First we give an example of a $(3,\epsilon)$-PR instance with at least one bad center,
and then we show that all $(3,\epsilon)$-PR instances must have at most 6 bad centers.

\begin{lemma}
For all $\alpha,n,k\geq 1$ such that $\frac{n}{k}\in \mathbb{N}$, there exists a clustering
instance with one bad center satisfying $\left(\alpha,\frac{2}{n}\right)$-perturbation resilience.
\end{lemma}

\begin{proof}
Given $\alpha,n,k\geq 1$, we construct a clustering instance such that all clusters are size 
$\frac{n}{k}$. Denote the clusters by $C_1,\dots,C_k$ and the centers by $c_1,\dots,c_k$.
For each $i$, denote the non-centers in $C_i$ by $p_{i,1},\dots,p_{i,L}$.
Now we define the distances as follows. For convenience, set $L=\frac{n}{k}-1$.
For all $2\leq i\leq k$ and $1\leq j\leq L$, let $d(c_i,p_{i,j})=1$.
For all $2\leq i\leq k$, $1\leq j,\ell\leq L$, let $d(p_{i,j},p_{1,\ell})=\frac{1}{\alpha}$
and $d(c_1,p_{1,\ell})=\frac{1}{\alpha}$.
Finally, let $d(p_{2,1},c_1)=1$.
All other distances are the maximum allowed by the triangle inequality.
In particular, the distance between two points $p$ and $q$ is set to infinity unless there
exists a path from $p$ to $q$ with finite distance edges defined above.
See Figure \ref{fig:badcenter}.

The optimal clusters and centers are $C_1,\dots,C_k$ and $c_1,\dots,c_k$, achieving a radius
of 1, and $c_1$ is a bad center because $d(p_{2,1},c_1)=1$ but $d(c_1,p_{2,1})=\infty$.
It is left to show that this instance satisfies $(\alpha,\frac{2}{n})$-perturbation resilience.
Given an arbitrary $\alpha$-perturbation $d'$, we must show that at most $\frac{2}{n}\cdot n=2$ points switch clusters. 
By definition of an $\alpha$-perturbation, for all $p,q$, we have $d(p,q)\leq d'(p,q)\leq \alpha d(p,q)$
(recall that without loss of generality, a perturbation only increases the distances).
The centers $c_2,\dots,c_k$ \emph{must} remain optimal centers under $d'$,
since for all $2\leq i\leq k$, $d'(c_i,p_{i,1})\leq \alpha$ and no other point
$q\neq c_i,p_{i,1}$ satisfies $d(q,p_{i,1})<\infty$.
Now we must determine the final optimal center.
Note that for all $2\leq i,j\leq k$ and $1\leq \ell,m\leq L$, we have
\begin{align*}
d'(p_{i,\ell},p_{1,m})&\leq \alpha d(p_{i,\ell},p_{1,m})\\
&\leq\alpha\cdot\frac{1}{\alpha}\\
&< d(c_j,p_{1,m})\\
&\leq d'(c_j,p_{1,m}).
\end{align*}

Therefore, $c_j$ cannot be a center for $p_{i,\ell}$, for all $2\leq i,j\leq k$ and $1\leq\ell\leq L$.
Therefore, the final optimal center $c$ under $d'$ must be either $c_1$ or $p_{i,\ell}$ for $2\leq i\leq k$ and $1\leq \ell\leq L$.
Furthermore, it follows that $c$'s cluster at least contains $C_1\setminus \{c_1\}$
and for each $2\leq i\leq k$, $c_i$'s cluster at least contains $C_i\setminus\{c\}$.
Therefore, the optimal clustering under $d'$ differs from $\OPT$ by at most two points.
This concludes the proof.
\end{proof}

\begin{figure}
\centering
\includegraphics[width =.6\textwidth]{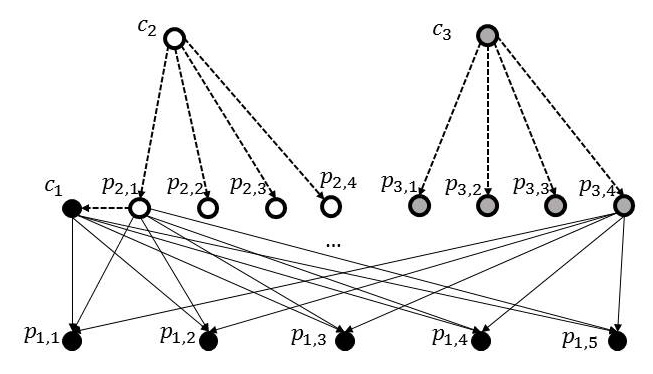}
\caption{\small \sl An $(\alpha,\epsilon)$-perturbation resilient asymmetric
	$k$-center instance with one bad center ($c_y$).
	The dotted arrows are distance 1, and the solid arrows are
	distance $\frac{1}{\alpha}$.}
\label{fig:badcenter}
\end{figure}

Now we show there are at most 6 bad centers for any asymmetric $k$-center instance satisfying
$(3,\epsilon)$-PR.

\begin{lemma} \label{lem:badcenters}
Given a $(3,\epsilon)$-perturbation resilient asymmetric $k$-center
instance such that all optimal clusters are size $>2\epsilon n$, there are at most 6 bad centers,
i.e., at most 6 centers $c_i$ such that $\exists q$ with $d(q,c_i)\leq r^*$ and $d(c_i,q)>r^*$.
\end{lemma}

\begin{proof} 
Assume the lemma is false. 
By assumption, there exists a set $B$, $|B|\geq 7$,
of centers $c_i$ such that $\exists q$ with $d(q,c_i)\leq r^*$ and $d(c_i,q)> r^*$.
The first step is to use this set of bad centers to construct a set $C$ of $\leq k-3$ points
which are $\leq 3r^*$ from every point in $S$.
Once we find $C$, we will show how this set cannot exist under $(3,\epsilon)$-perturbation resilience, causing a contradiction.

Given a center $c_i\in B$, and $q$ such that $d(q,c_i)\leq r^*$ and $d(c_i,q)> r^*$,
note that $d(c_i,q)>r^*$ implies $q\notin C_i$.
For each $c_i\in B$, define $a(i)$ as the center of $q$'s cluster.
Then $d(a(i),c_i)\leq d(a(i),q)+d(q,c_i)\leq 2r^*$ and so for all $p\in C_i$,
we have $d(a(i),p)\leq d(a(i),c_i)+d(c_i,p)\leq 3r^*$.
If for each $c_i\in B$, $a(i)$ is not in $B$, then we would be able to remove $B$ from the set of optimal centers,
and the remaining centers are still distance $3r^*$ from all points in $S$ (finishing the first half of the proof).
However, we need to consider the case where there exist centers $c_i$ in $B$ such that $a(i)$ is also in $B$.
Our goal is to show there exists a subset $B'\subseteq B$ of size 3, such that for each $c_i\in B'$, 
$a(i)\notin B'$, therefore, the set of optimal centers without $B'$ is still distance $3r^*$ from all points in $S$.

Construct a directed graph $G=(B,E)$ where $E=\{(c_i,c)\mid c=a(i)\}$.
Then every point has out-degree $\leq 1$.
Finding $B'$ corresponds to finding $\geq 3$ points with no edges to one another,
i.e., an independent set of $G$.
Consider a connected component $G'=(V',E')$ of $G$. 
Since $V'$ is connected, we have $|E'|\geq |V'|-1$. 
Since every vertex has out-degree $\leq 1$, $|E'|\leq |V'|$. Then we have two cases.

Case 1: $|E'|=|V'|-1$. Then $G'$ is a tree, 
and so there must exist an independent set of size $\left\lceil \frac{|V'|}{2}\right\rceil$.

Case 2: $|E'|=|V'|$. Then $G'$ contains a cycle, and so there exists an independent set of size 
$\left\lfloor \frac{|V'|}{2}\right\rfloor$.

It follows that we can always find an independent set of size $\left\lfloor \frac{|V'|}{2}\right\rfloor$
for the entire graph $G$. 
For $|B|\geq 7$, there exists such a set $B'$ of size $\geq 3$.
Then we have the property that $c_i\in B'\implies a(i)\notin B'$.

Now let $C=\{c_\ell\}_{\ell=1}^k\setminus B'$. 
By construction, $B'$ is distance $\leq 3r^*$ to all points in $S$.
Consider the following 3-perturbation $d'$: increase all distances by a factor of 3, except $d(a(i),p)$,
for $i$ such that $c_i\in B'$ and $p\in C_i$, which we increase to $\min(3r^*,3d(a(i),p))$.
Then by Lemma \ref{lem:d'_same_cost}, the optimal radius is $3r^*$. Therefore, the set $C$
achieves the optimal cost over $d'$ even though $|C|\leq k-3$.
Then we can pick any combination of 3 dummy centers, and they must all result in clusterings which are
$\epsilon$-close to $\OPT$. We will show this contradicts $(3,\epsilon)$-perturbation resilience.

We pick five arbitrary points $p_1,p_2,p_3,p_4,p_5\in S\setminus C$,
and define $C'=C\cup\{p_1,p_2,p_3,p_4,p_5\}$.
From the above paragraph, each size 3 subset $P\subseteq\{p_1,p_2,p_3,p_4,p_5\}$ added to $C$ will result
in a set of optimal centers under $d'$.
Then by Fact \ref{fact:half}, each point in $C\cup P$ must be the center for the majority of points in exactly one cluster.
To obtain a contradiction, we consider the ranking defined by Fact \ref{fact:ranking} of $C'$ over $d'$.

We we start with a claim about the rankings: for each $c'\in C'$, for all pairs $x,y$ such that $x\neq y$,
if $\nexists c\in C$ such that $R_{x,d',C'}(c)<R_{x,d',C'}(c')$ or $R_{y,d',C'}(c)<R_{y,d',C'}(c')$,
then $R_{x,d',C'}(c')+R_{y,d',C'}(c')\geq 5$.
In words, there cannot be two clusters such that $c'$ is ranked first among $C\cup\{c'\}$ and top two (or first and third) among $C'$ for both clusters.
Assume this is false. Then there exist $x\neq y$ such that $R_{x,d',C'}(c')+R_{y,d',C'}(c')\leq 4$,
so there are at most two total points ranked above $c'$ in $R_{x,d',C'}$ and $R_{y,d',C'}$, and these points must be from the set $\{p_1,p_2,p_3,p_4,p_5\}$.
Without loss of generality, denote these points by $p$ and $p'$ (if there are one or zero points ranked above $c'$, let one or both of $p$ and $p'$ be arbitrary).
Then consider the set of centers $C'\setminus \{p,p'\}$ which is size $k$ and must be optimal under $d'$ as described earlier.
However, the partitioning is not $\epsilon$-close to $\OPT$, since $c'$ is the best center (ranked 1) for both $C_x$ and $C_y$.
This completes the proof of the claim.

Now consider the set $D=\{c_i\in C\mid \exists x\text{ s.t. }R_{x,d',C'}(c_i)=1\}$, i.e., the set of points in $C$ which are ranked 1 for some cluster.
Denote $m=(k-3)-|D|$, which is the number of points in $C$ which are \emph{not} ranked 1 for any cluster.
By the claim and since $|C|=k-3$, there are exactly $m+3$ clusters whose top-ranked point is not in $C$.
Given one such cluster $C_x$, again by the claim, the top two ranked points must not be from the set $D$.
Therefore, there are $2(m+3)$ slots that must be filled by $m+5$ points, so (for all $m\geq 0$) by the Pigeonhole Principle, there must exist
a point $p\in C'$ ranked in the top two by two different clusters. This directly contradicts the claim, so we have a contradiction which completes the proof.
\end{proof}

\subsubsection{Algorithm under $(3,\epsilon)$-PR}

From the previous lemma, we know that at most a constant number of centers are bad.
Essentially, our algorithm runs a symmetric 2-approximation algorithm on $A$, for all $k-6\leq k'\leq k$,
to find a 2-approximation for the clusters in $A$.
For instance, iteratively pick an unmarked point, and mark all points distance $2r^*$ away from it
\citep{hochbaum1985best}.
Then we use brute force to find the remaining 6 centers,
which will give us a 3-approximation for the entire point set.
Under $(3,\epsilon)$-perturbation resilience, this 3-approximation must be 
$\epsilon$-close to $\OPT$.
We are not able to output $\OPT$ exactly, since
Condition 2 may not be satisfied for up to $\epsilon n$ points.
The asymmetric $k$-center algorithm runs an approximation algorithm for symmetric $k$-center as a subroutine.
The symmetric $k$-center instance $(S,A,d)$ is a generalization: the set of allowable centers $A$ is a subset of
the points $S$ to be clustered. The classic 2-approximation algorithms for $k$-center apply to this setting as well.

\begin{algorithm}[h]
\caption {\textsc{$(3,\epsilon)$-Perturbation Resilient Asymmetric $k$-center}}
\label{alg:3eps-asy}
\begin{algorithmic} 
\Require{Asymmetric $k$-center instance $(S, d)$, $r^*$ (or try all possible candidates).}
\begin{enumerate}[leftmargin=*]
\item Build set $A=\{p\mid \forall q, d(q,p)\leq r^* \implies d(p,q)\leq r^*\}$.
\item Create the threshold graph $G=(A,E)$ where $E=\{(u,v)\mid d(u,v)\leq r^*\}$.
Define a new symmetric $k$-center instance $(S,A,d')$ where $d'(u,v)=dist_G(u,v)$.
\item For all $k-6\leq k'\leq k$, run a symmetric $k$-center 2-approximation algorithm on $(S,A,d')$.
Break if the output is a set of centers $C$ achieving cost $\leq 2r^*$.
\item 
For all $C'\subseteq C$ of size $k-6$ and $S'\subseteq S$ of size 6, return if
cost$(C'\cup S')\leq 3r^*$.
\end{enumerate}
\Ensure{Voronoi tiling $G_1, \dots, G_k$
using $C'\cup S'$ as the centers.}
\end{algorithmic}
\end{algorithm}

\begin{theorem} \label{thm:3eps_asy}
Algorithm \ref{alg:3eps-asy} runs in polynomial time and outputs a clustering 
that is $\epsilon$-close
to $\OPT$, for $(3,\epsilon)$-perturbation resilient asymmetric $k$-center
instances such that all optimal clusters are size $>2\epsilon n$.
\end{theorem}

\begin{proof}
We define three types of clusters. A cluster $C_i$ is green if $c_i\in A$, it is yellow if $c_i\notin A$ but $C_i\cap A\neq\emptyset$,
and it is red if $C_i\cap A=\emptyset$. Denote the number of yellow clusters by $y$, and the number of red clusters by $x$.
From Lemma \ref{lem:badcenters}, we know that $x+y\leq 6$.
The symmetric $k$-center instance $(S,A,d')$ constructed in step 2 of the algorithm is a subset of an instance with $k-x$ optimal clusters
of cost $r^*$, so the $(k-x)$-center cost of $(S,A,d')$ is at most $r^*$. Therefore, step 3 will return a set of centers achieving cost $\leq 2r^*$
for some $k'\leq k-x$.
By definition of green clusters, we know that $k-x-y$ clusters have their optimal center in $A$.
For each green cluster $C_i$, let $c(i)\in C$ denote the center which is distance $\leq 2r^*$
to $c_i$ (if there is more than one point in $C$, denote $c(i)$ by one of them arbitrarily).
Let $C'=\{c(i)\mid C_i\text{ is green}\}$, and $|C'|\leq k-x-y$.
Then the set $C'\cup \{c_x\mid x\text{ is not green}\}$ is cost $\leq 3r^*$, and the
algorithm is guaranteed to encounter this set in the final step.

Finally, we explain why $C'\cup \{c_x\mid x\text{ is not green}\}$ must be $\epsilon$-close to $\OPT$.
Let $B=\{c_x\mid x\text{ is not green}\}$.
Create a 3-perturbation in which we increase all distances by 3, except
for the distances from $C'\cup B$ to all points in their Voronoi tile, which we
increase up to $3r^*$. Then, the optimal score is $3r^*$ by Lemma
\ref{lem:d'_metric}, and $C'\cup B$ achieves this score. Therefore, 
by $(3,\epsilon)$-perturbation resilience, the Voronoi
tiling of $C'\cup B$ must be $\epsilon$-close to $\OPT$.
This completes the proof.
\end{proof}

\subsection{APX-Hardness under perturbation resilience}
Now we show hardness of approximation even when it is guaranteed the clustering satisfies 
$(\alpha,\epsilon)$-perturbation resilience for $\alpha\geq 1$ and $\epsilon>0$. 
The hardness is based on a reduction from the general clustering
instances, so the APX-hardness constants match the non-stable APX-hardness results.
This shows the condition on the cluster sizes in Theorem \ref{thm:3epsThm} is tight.
\footnote{
In fact, this hardness holds even under the strictly stronger notion of 
\emph{approximation stability}~\citep{as}, therefore, it generalizes a hardness result by \cite{as}.
}

\begin{theorem} \label{thm:apx}
Given  $\alpha\geq 1$, $\epsilon>0$,
it is NP-hard to approximate $k$-center to 2, $k$-median to $1.73$, or $k$-means to 3.94,
even when it is guaranteed the instance satisfies $(\alpha,\epsilon)$-perturbation resilience.
\end{theorem}

\begin{proof}
Given $\alpha\geq 1$, $\epsilon>0$, assume there exists a $\beta$-approximation algorithm $\mathcal{A}$ for $k$-median under
$(\alpha,\epsilon)$-perturbation resilience. We will show a reduction to $k$-median without perturbation resilience.
Given a $k$-median clustering instance $(S,d)$ of size $n$, we will create a new instance $(S',d')$ for $k'=k+n/\epsilon$ with
size $n'=n/\epsilon$ as follows.
First, set $S'=S$ and $d'=d$, and then add $n/\epsilon$ new points to $S'$, such that their distance to every other point is
$2\alpha n\max_{u,v\in S}d(u,v)$. Let $\OPT$ denote the optimal solution of $(S,d)$. Then the optimal solution to $(S',d')$ is
to use $\OPT$ for the vertices in $S$, and make each of the $n/\epsilon$ added points a center. 
Note that the cost of $\OPT$ and the optimal clustering for $(S',d')$ are identical, since the added points are distance 0 to their center.
Given a clustering $\mathcal{C}$ on $(S,d)$, let $\mathcal{C}'$ denote the clustering of $(S',d')$ that clusters $S$ as in $\mathcal{C}$, and then adds
$n/\epsilon$ extra centers on each of the added points. Then the cost of $\mathcal{C}$ and $\mathcal{C}'$ are the same, so it follows that
$\mathcal{C}$ is a $\beta$-approximation to $(S,d)$ if and only if $\mathcal{C}'$ is a $\beta$-approximation to $(S',d')$.
Next, we claim that $(S',d')$ satisfies $(\alpha,\epsilon)$-perturbation resilience.
Given a clustering $\mathcal{C}'$ which is an $\alpha$-approximation to $(S',d')$, then there must be a center located at all
$n/\epsilon$ of the added points, otherwise the cost of $\mathcal{C}'$ would be $>\alpha\OPT$.
Therefore, $\mathcal{C}'$ agrees with the optimal solution on all points except for $S$, therefore, $\mathcal{C}'$ must be $\epsilon$-close to
the optimal solution.
Now that we have established a reduction, the theorem follows from hardness of $1.73$-approximation for $k$-median~\citep{jain2002new}.
The proofs for $k$-center and $k$-means are identical, using hardness from~\citep{gonzalez1985clustering} and~\citep{jain2002new}, respectively.
\end{proof}

\section{Conclusion} \label{sec:conclusion}

Our work pushes the understanding of (promise) stability conditions farther in several ways. 
We are the first to design computationally efficient algorithms to find the optimal clustering under
$\alpha$-perturbation resilience
with a constant value of $\alpha$ for a problem that is hard to approximate to any constant factor in the worst
case, thereby demonstrating the power of perturbation resilience. Furthermore, we demonstrate the limits of
this power by showing the first tight results in this space for perturbation resilience.
Our work also shows a surprising relation between symmetric and asymmetric instances, in that they
are equivalent under resilience to 2-perturbations, which is in stark contrast to their widely differing tight
approximation factors.
Finally, we initiate the study of clustering under local stability.
We define a local notion of perturbation resilience, and
we give algorithms that simultaneously output all optimal clusters satisfying local stability, while
ensuring the worst-case approximation guarantee.
Although $\alpha=2$ is tight for $k$-center, the best value of perturbation resilience for symmetric $k$-median and
other center-based objectives is not known. Currently, the best upper bound is $\alpha=2$ \citep{angelidakis2017algorithms},
but no lower bounds are known for $\alpha=1+\epsilon$, for constant $\epsilon>0$.

\section*{Acknowledgments} \label{sec:acks}
This work was supported in part by NSF grants
CCF 1535967, CCF-1422910, CCF-145117, IIS-1618714, a Sloan Research
Fellowship, a Microsoft  Research Faculty  Fellowship,  a Google
Research Award,  an IBM Ph.D. fellowship,  a National Defense Science
\& Engineering Graduate  (NDSEG) fellowship,
and by the generosity of Eric and Wendy Schmidt by recommendation of the Schmidt Futures program.

\newpage


\bibliographystyle{plainnat}
\bibliography{clustering}







\end{document}